\documentclass[12pt,a4paper,oneside]{article}
\usepackage{latexsym,amsthm,amsmath,amssymb,verbatim,xcolor,natbib,graphicx,placeins,rotating}
\usepackage{algorithm,bbold}
\usepackage{algpseudocode}
\usepackage{sidecap}
\usepackage{bm}
\usepackage[margin=1in]{geometry}
\usepackage{caption}
\usepackage{subcaption}
\usepackage{authblk}
\usepackage{epstopdf}
\theoremstyle{plain}

\newtheorem{proposition}{Proposition}

\theoremstyle{definition}
\newtheorem{definition}{Definition}

\newtheorem*{example}{Example}

\theoremstyle{remark}

\usepackage{algpseudocode}
\usepackage{algorithm}
\usepackage{epsfig}
\usepackage{ulem} 
\usepackage{xspace} 
\usepackage{graphicx}
\usepackage{times,epsfig,makeidx,amsfonts,amsmath}
\usepackage{lscape} 
\graphicspath{{figures/}} 

\newcommand{\E}{\text{E}} 
\newcommand{\Var}{\text{Var}} 

\begin{document}

\title{Sequential importance sampling for multi-resolution Kingman-Tajima coalescent counting}

\author[*]{Lorenzo Cappello}
\author[*]{Julia A. Palacios} 
\affil[*]{Stanford University}

\maketitle



\abstract{Statistical inference of evolutionary parameters from molecular sequence data relies on coalescent models to account for the shared genealogical ancestry of the samples. However, inferential algorithms do not scale to available data sets. A  strategy to improve computational efficiency is to rely on simpler coalescent and mutation models, resulting in smaller hidden state spaces. An estimate of the cardinality of the state-space of genealogical trees at different resolutions is essential to decide the best modeling strategy for a given dataset. To our knowledge, there is neither an exact nor approximate method to determine these cardinalities. We propose a sequential importance sampling algorithm to estimate the cardinality of the sample space of genealogical trees under different coalescent resolutions. Our sampling scheme proceeds sequentially across the set of combinatorial constraints imposed by the data, which in this work are completely linked sequences of DNA at a non recombining segment. We analyze the cardinality of different genealogical tree spaces on simulations to study the settings that favor coarser resolutions. We apply our method to estimate the cardinality of genealogical tree spaces from mtDNA data from the 1000 genomes and a sample from a Melanesian population at the $\beta$-globin locus.}

\section{Introduction}
\label{intro}
Statistical inference of evolutionary parameters, such as effective population size $N(t)$, from molecular sequence data is an important task in population genetics, conservation biology, anthropology and public health \citep{nor98,ros02,liu13}. Inference of such parameters relies on the coalescent process that explicitly models the shared ancestry of a sample (genealogy) of $n$ individuals from a population. More specifically, in the standard neutral coalescent framework, observed molecular data $\mathbf{Y}$ at a non-recombining segment from a sample of $n$ individuals within a population, is the result of a point process of mutations with rate $\mu$ superimposed on the genealogy $\mathbf{g}$ of the sample. The genealogy itself is not directly observed but it is assumed to be a realization of a stochastic ancestral process (coalescent process) that depends on $N(t)$. Figure \ref{fig:coalescent} shows a realization of the standard coalescent (genealogy) and mutations. 

\begin{SCfigure}[2][h]
	\hspace{-0.2cm}\includegraphics[width=.6\textwidth]{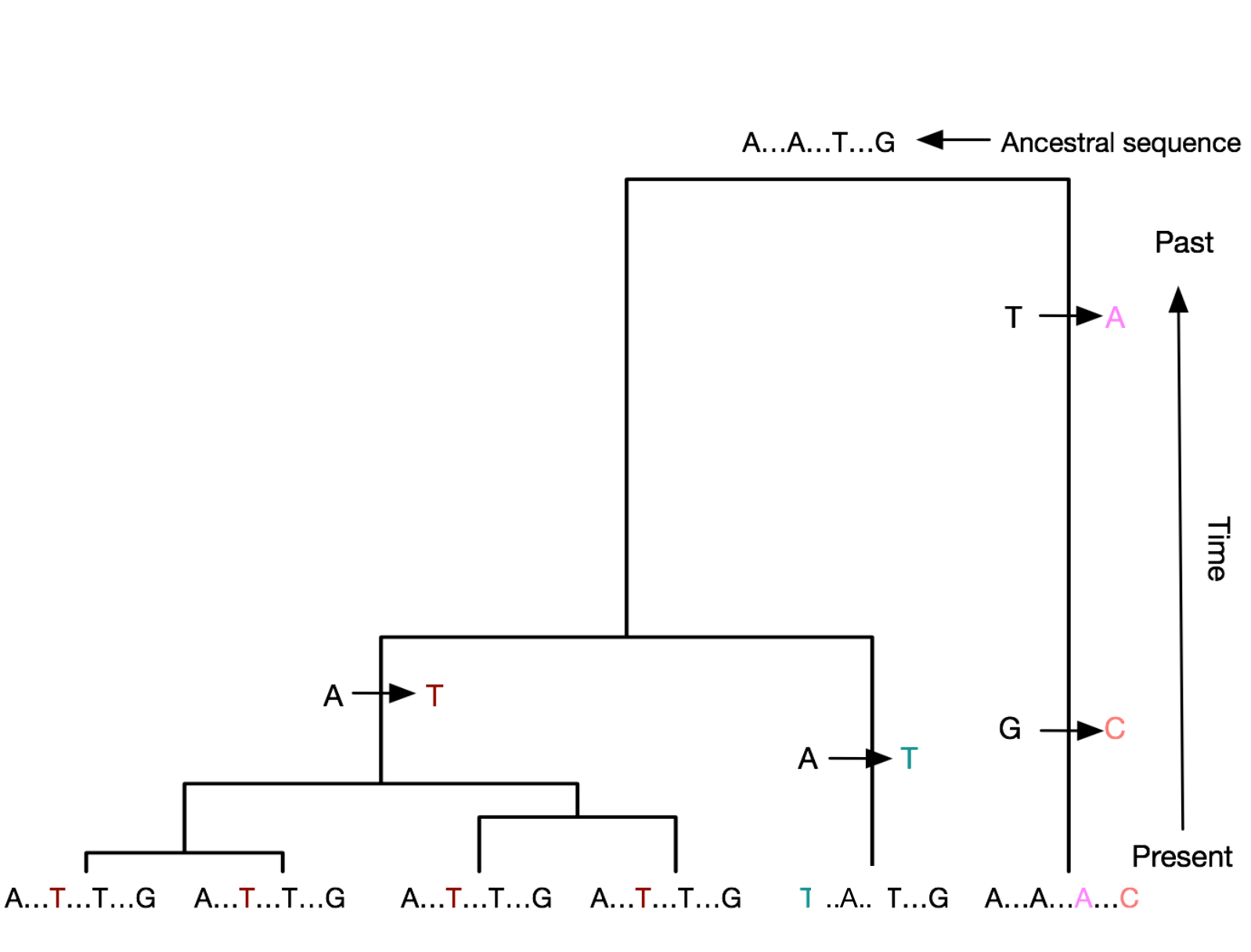}
	\caption{\small{\textbf{Coalescence and mutation.} A genealogy of 6 individuals at a locus of 100 base pairs is depicted as a bifurcating tree. 
			Four mutations (at different sites) are superimposed along the branches of the tree giving rise to the 6 sequences shown at the tips of the tree. The 96 sites (base pairs) that do not mutate are represented by dots and only the nucleotides at the polymorphic sites are shown.}}
	\label{fig:coalescent}
\end{SCfigure}

Both Bayesian and frequentist methods rely on the marginal likelihood calculated by integrating over the latent space of genealogies, that is:
\begin{equation}
\label{eq:integ}
P(\textbf{Y}| N(t), \mu)= \int_{\mathbf{g} \in \mathcal{G}\times \mathbb{R}^{n-1}} P(\textbf{Y}\mid \mathbf{g}, \mu) P(\mathbf{g}\mid N(t)) d\mathbf{g}.\end{equation}
Integration in the previous equation involves the sum over all possible tree topologies and $n-1$ integrals over coalescent times $\mathbf{t}\in \mathbb{R}^{n-1}$ (bifurcating times).  The integral in \eqref{eq:integ} is usually approximated via Monte Carlo (MC) or Markov chain Monte Carlo (MCMC). 
However, the cardinality of the hidden state space of tree topologies $|\mathcal{G}|$ grows superexponentially with the number of samples $n$, making integration over the space of genealogies already challenging for small $n$. 

In order to gain computational tractability, 
several methods have been proposed to infer $N(t)$ from summary statistics such as the site frequency spectra \citep{TerhorstSFS}, from an estimated genealogy \citep{pal13,Gattepaillegenetics.115.185058}, or from a small number of samples \citep{drummond2012bayesian}. \citet{Gao235} present an extensive list of implemented methods.

Alternative approaches that rely on lower resolution coalescent models have been recently proposed  \citep{sai15,sai18, pal18}. The appealing advantage of these approaches is the \textit{a priori} drastic reduction in the cardinality of the space of tree topologies for a fixed $n$. 
However, conditionally on a given dataset and the popular \textit{infinite sites} mutation model \citep{wat75},  the true reduction in cardinality, that is, the number of tree topologies for which $P(\mathbf{Y}\mid \mathbf{g},\mu)>0$ (compatible), is not known neither analytically nor approximately. 

In this work, we propose a set of algorithms to approximate the cardinality of different tree topology spaces modeled at different coalescent resolutions, the so-called Kingman-Tajima resolutions \citep{sai15}. Reliable estimation of the cardinality of the coalescent hidden state space should provide a valuable guidance to statisticians in designing methods employing these different resolutions. State-space count also offers an important auxiliary tool for practitioners: first, it can aid tuning parameters of the MCMC chains-- e.g. length of the chain; second, and closely related, it is informative of the computational feasibility of coalescent based inference for a given dataset -- e.g. we will quantify how it is not solely sample size that drives computational feasibility  but, for fixed $n$, the state space size, and consequently the computational burden, varies largely as a function of the data at hand; lastly, it may offer a convergence diagnostic criteria for sampling methods -- e.g. what proportion of the state space has been explored in the approximate posterior distribution.
In addition, the cardinality of the topological tree space is already an input of some inferential algorithms beyond MCMC, such as the combinatorial sequential Monte Carlo \citep{wan15}. 

Counting genealogical trees is a very active area of research in biology and mathematics starting from \citet{cay56}; see \citet{ste16} for a review. To our knowledge, the large body of work in this area has focused on \textit{exact} combinatorial results or \textit{recursive} algorithms to explore a constrained space. In this work, the combinatorial question of counting the number of compatible tree topologies with the data is treated as a \textit{statistical} problem: estimation of the normalizing constant of a uniform discrete distribution over the space of compatible tree topologies \citep{jer86}. In this work we estimate the normalizing constant by sampling compatible trees.

Lacking a trivial uniform sampling algorithm in this context, there are two classes of methods to estimate the cardinality of discrete structures subject to constraints: MCMC and sequential importance sampling (SIS). There is a large literature documenting both applications to challenging combinatorial problems and good empirical performances of both MCMC methods \citep{jer96, bla09, sin12}, and SIS methods \citep{knu76,chen05,bli11,chen18,dia18}. However, there is not a prevailing consensus that one method outperforms the other, even within the same application. Moreover, we are not aware of the use of these methods in the context of coalescent models. \\
Our estimation method is an instance of SIS. More specifically, our algorithm sequentially samples topologies $g$ compatible with the data with a tractable sampling probability $q(g)$. The SIS estimation of the cardinality is computed by a Monte Carlo approximation of the following expectation:
\begin{equation}
\label{eq:count1}
\E_q\left[\frac{1}{q(g)}\right]= \sum_{g \in \mathcal{G}_C} \frac{1}{q(g)} q(g)= |\mathcal{G}_C|,
\end{equation}
where $\mathcal{G}_C$ is the space of compatible tree topologies. The main contribution of this work is a set of algorithms that sample only compatible tree topologies under different coalescent models, and consequently correspond to different proposals $q$. Whereas the focus of this work is the estimation of state-space cardinalities,  it is easy to see that the same procedure can be applied to enumerate tree topologies with certain features of interest. For example, the number of balanced trees and trees with certain shape are indicatives of  population structure and phylogenetic diversity \citep{Ferretti229,Maliet2018}, the number of cherries and pitchforks are indicatives of neutrality \citep{grif87,di13}. 
Although we do not explore this research direction in this paper, our algorithms offer a building block to study how a neutral coalescent model fits the data set at hand.

The rest of the paper proceeds as follows. Section 2 reviews the Kingman-Tajima coalescent and the perfect phylogeny representation of molecular sequence data. In Section 3, we present the sampling algorithms.  In section 4 we analyze the cardinality of genealogical spaces under different coalescent resolutions from simulated data and in section 5 we present two case studies, one case study from simulated data and one case study of a sample of human mtDNA from the 1000 genomes and from other human DNA datasets. Section 6 concludes. 
%


\section{Preliminaries}
\label{prel}

\subsection{Kingman-Tajima coalescent}

%

\textit{Kingman's coalescent}  is a continuous-time Markov chain 
with state space the set of partitions of the label set $[n]= \{1, \ldots, n\}$ of the n individuals in a sample \citep{kingn82}. The process starts at $\{\{1\},\ldots,\{n\}\}$, it then jumps when two of the n
individuals coalesce (represented as the merger of two branches in a single internal node in the genealogy). The state of the process after the first transition is the partition of $[n]$ into $n-1$ sets, one set with the labels of the two individuals that coalesce and $n-1$ singleton sets with the labels of the remaining individuals. The process ends when all individuals coalesce, i.e. at state $\{1,\ldots,n\}$ when there is a single set (at the root of the genealogy when all individuals have a common ancestor). 

A complete realization of Kingman's coalescent process is commonly represented as a timed bifurcating tree (genealogy) denoted by $\mathbf{g}^{K}=\{g^{K},\mathbf{t}\}$.  In this work we concern ourselves with the tree topology only, \textit{i.e.} a complete realization of the embedded jump chain of the process $g^{K}=\{c_{i}\}_{i=n:1}$, where $c_i$ is the state of the process when there are $i$ branches. A genealogical representation of $g^{K}$ is given in Figure~\ref{fig:tree_bijection}(b) and the corresponding chain in Figure~\ref{fig:tree_bijection}(a). Superindex $K$ in $g^{K}$ serves to distinguish a Kingman's tree topology to any other type of tree topology. The transition probability of the jump chain is:

\begin{equation}
\label{king_move}
P(c_{i-1} \mid c_{i})=\left\{
\begin{array}{ll}
\binom{i}{2}^{-1} \,\ \,\ \,\ \,\ \text{if}  \ \,\ c_{i-1} \prec c_i\\
0 \,\ \,\ \,\ \,\  \,\ \,\ \,  \,\ \,\  \text{otherwise}
\end{array}
\right.
\end{equation}
where $c_{i-1} \prec c_i$ means that $c_{i-1}$ can be obtained from joining two elements of $c_i$. 
It follows from \eqref{king_move} that $P(g^K)= 2^{n-1}/[n! (n-1)!],$
i.e. the discrete uniform over all possible chain trajectories. We will use $\mathcal{G}_n^K$ to denote the space of such Kingman's topologies.
\begin{figure}[!htbp]
	\includegraphics[clip, width=1\textwidth]{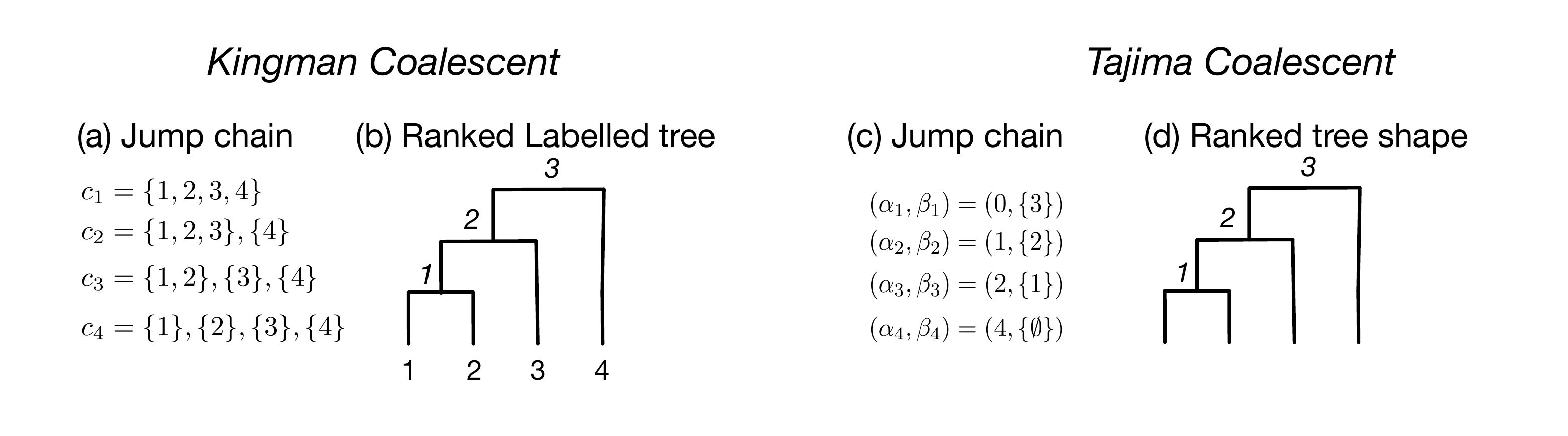}
	\caption{\textbf{Coalescent tree topologies.} \textbf{(a)} A complete realization from Kingman's jump chain, and \textbf{(b)} its corresponding bijection: a ranked labeled tree topology. \textbf{(c)} A complete realization from Tajima's jump chain, and \textbf{(d)} its corresponding bijection: a ranked tree shape.}\label{fig:tree_bijection}
\end{figure}

\textit{Tajima's coalescent}  is a continuous-time Markov chain whose complete realization is also in bijection with a timed bifurcating tree. Its embedded jump chain $\{(\alpha_{i},\beta_{i})\}_{i=n:1}$ keeps track of the number of singletons $\alpha_i$ and the set of extant vintage labels $\beta_{i}$ when there are $i$ branches \citep{taj83,sai15}. We refer to singleton branch as a branch in the tree that subtends a leaf, and a vintage as the internal branch that subtends the subtree labeled by the ranking at which the subtree was created in the jump chain.  Since singletons' labels are ignored, there are up to three types transitions: two singletons merge, one singleton and a vintage merge, or two vintages merge. Formally, given a current state $(\alpha_j, \beta_j)$, when there are $j=\alpha_j+|\beta_{j}|$ branches in the genealogy, the chain transitions to $\alpha_{j-1}=\alpha_{j}-2$ and $\beta_{j-1}=\beta_{j} \cup \{j\}$ if two singletons create a new vintage node with label $\{j\}$; the chain transitions to $\alpha_{j-1}=\alpha_{j}-1$ and $\beta_{j-1}=\beta_{j} \setminus \{i\} \cup \{j\}$ if one singleton and vintage branch with label $\{i\}$ merge to create a new vintage node with label $\{j\}$; and the chain transitions to $\alpha_{j-1}=\alpha_{j}$ and $\beta_{j-1}=\beta_{j} \setminus \{i,k\}\cup \{j\}$ if vintages $\{i\}$ and $\{k\}$ merge to create a new vintage with label $\{j\}$. 
The process starts at  state $\alpha_n=n$ and $\beta_n=\emptyset$ (at the tips of the tree). The chain then jumps to $\alpha_{n-1}=n-2$ and $\beta_{n-1}=\{1\}$ (with probability one since this is the only possible transition at this step), and the vintage $\{1\}$ is created. The process ends at the root when there is a single vintage, i.e. $\alpha_{1}=0$, and $\beta_{1}=\{n-1\}$. 
A complete realization of Tajima's coalescent continuous process can be represented as a genealogy $\mathbf{g}^{T}=\{g^{T},\mathbf{t}\}$. A complete realization of the jump chain of the process is denoted by $g^{T}=\{(\alpha_{i},\beta_{i})\}_{i=n:1}$ (Figure \ref{fig:tree_bijection}(c)). The jump chain has the following transition probabilities:

\begin{eqnarray}
\label{eq:taj_move}
P\Big[(\alpha_{i-1},\beta_{i-1})\Big| (\alpha_{i},\beta_{i})\Big]=\begin{cases}
\frac{\binom{\alpha_{i}}{\alpha_{i}-\alpha_{i-1}}}{\binom{\alpha_{i}+|\beta_{i}|}{2}} & \text{if \ \ \  }  (\alpha_{i-1},\beta_{i-1}) \prec (\alpha_{i},\beta_{i})\\
0 & \text{otherwise}\\
\end{cases},
\end{eqnarray}
Given \eqref{eq:taj_move}, one can compute the probability of a Tajima's tree topology $g^T$ as $P(g^T)=2^{n-c(g)-1}/(n-1)!,$ where $c(g)$ is the number of coalescent events joining two singletons. We will use $\mathcal{G}_n^T$ to denote the space of such Tajima's topologies.

While Kingman's coalescent keeps track of who is related to whom, Tajima's coalescent describes the evolutionary relationships of a sample of $n$ individuals by keeping track of the number of singletons and the vintage labels of extant ``families". We note that Tajima's coalescent has the same number of transitions and wait time distribution as Kingman's coalescent. 
Tajima's coalescent is a lower-resolution coalescent process since it takes values in a smaller state-space than Kingman's. \citet{sai15} formalize this notion and describe in detail other coalescent resolutions. 


The corresponding tree topology under Kingman coalescent $g^{K}$ is a \textit{ranked labeled tree} and the corresponding tree topology under Tajima coalescent $g^{T}$ is a \textit{ranked tree shape} (Figure \ref{fig:tree_bijection}). The formal definitions are as follows:
\begin{definition}
	A \textit{ranked labeled tree} is a rooted binary tree with unique labels at the tips and a total ordering (ranking) for the internal nodes.
\end{definition}

\begin{definition}
	A \textit{ranked tree shape} is a rooted binary unlabeled tree with a total ordering (ranking) for the internal nodes.
\end{definition}

Although our main objective is to analyze Kingman and Tajima tree topologies, we extend our analysis to the corresponding unranked tree topologies: unranked labeled tree and tree shapes. Figure \ref{fig:tree_topologies}  shows the four tree topologies analyzed in this manuscript.

\begin{figure}[!htbp]
	\includegraphics[clip, width=1\textwidth]{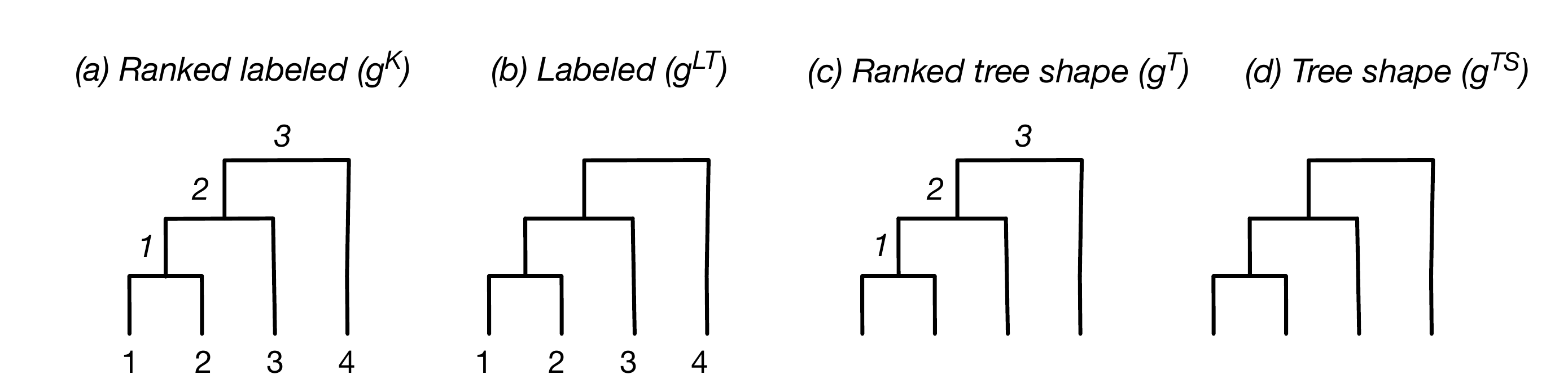}
	\caption{\small{\textbf{Tree topologies:} The \textbf{(a)} ranked labeled tree topology ($g^K$, Kingman), \textbf{(b)} labeled (unranked) tree topology ($g^{LT}$), \textbf{(c)} ranked tree shape ($g^T$, Tajima) and \textbf{(d)} tree shape ($g^{TS}$)}}
	\label{fig:tree_topologies}
\end{figure}

There are explicit or recursive formulas to compute the number of topologies with $n$ leaves.  The number of ranked labeled trees is 
$|\mathcal{G}_n^K|=n! (n-1)! /2^{n-1}$; the number of unranked labeled trees (binary phylogenetic trees) is $|\mathcal{G}_n^{LT}|=(2n-3)!!$ \citep{ste16}; 
the number of ranked tree shapes $|\mathcal{G}_n^T|$ is the $n$-th term of the Euler zig-zag sequence (alternating permutations, OEIS: A000111) \citep{di13}, and the number of tree shapes is the $n$-th Wedderburn-Etherington number (OEIS: 01190)  \citep{ste16}.  

For $n>3$, it holds that $|\mathcal{G}_n^{TS}| < |\mathcal{G}_n^{LT}|$ and $|\mathcal{G}_n^T| <|\mathcal{G}_n^K|$, that is the unlabeled tree topologies have smaller cardinalities that the labeled counterparts. For example, for $n=5$, there are $180$ ranked labeled trees and $5$ unlabeled ranked trees. Similarly, $105$ labeled trees and $3$ tree shapes. This cardinality difference has motivated the study of lower resolution coalescent processes \citep{sai15}. However, it is not clear how big this difference is when the observed data restricts the  space of topologies. In the next section, we describe how observed data imposes combinatorial constraints on the topological space.

\subsection{Perfect phylogeny and infinite sites  model}
\label{perphylo}
As mentioned in the introduction, we assume that molecular variation at a non-recombining contiguous segment of DNA (or locus) is the result of a mutation process superimposed on the timed genealogy $\mathbf{g}$ (Figure \ref{fig:coalescent}). Here, we assume that mutations (or substitutions) occur at sites that have not mutated previously. This mutation model is called the \textit{infinite-sites model} (ISM) \citep{kim69,wat75}. Further, we assume that our data consists of a single non-recombining segment of DNA. Although we will not model the mutation process explicitly, it is commonly assumed that mutation happens as Poisson process on the timed genealogy $\mathbf{g}$. However, an important consequence is that the ISM imposes a restriction on the space of tree topologies: given that at most one mutation occurs at a site, this mutation must occur on a branch subtending individuals with the observed mutation (Figure \ref{fig:coalescent}). Therefore mutations partition the observed sequences into two sets: the sequences that carry the mutations and the sequences that do not. In addition, if the ancestral type at each polymorphic site is known, molecular data from $n$ individuals at $m$ polymorphic sites can be represented as an incidence matrix $\textbf{Y}$ and a vector of the row frequencies of the matrix $\textbf{Y}$. The incidence matrix $\textbf{Y}$ is a $k \times m$ matrix with $0$-$1$ entries, where $0$ indicates the ancestral type and $1$ indicates the mutant type; $k$ is the number of unique sequences (or haplotypes) observed in the sample and the vector of frequencies indicates the number of times each haplotype is observed in the sample. For example, the $n=6$ sequences displayed at the leaves of the genealogy in Figure \ref{fig:coalescent} can be summarized as the incidence matrix and corresponding frequency vector in Figure \ref{fig:perfect_phylo}. The three haplotypes in this example are A...A...A...C, T...A...T...G and A...T...T...G with labels $h_{a},h_{b}$ and $h_{c}$ respectively. In this example, the ancestral sequence is displayed at the root of the tree in Figure \ref{fig:coalescent}. In what follows, we will assume that our data are 
an incidence matrix and corresponding frequencies as in Figure \ref{fig:perfect_phylo}. 



\citet{gus91} proposed an algorithm to represent the incidence matrix as a multifurcating tree called \textit{perfect phylogeny}. A perfect phylogeny is in bijection with an incidence matrix and it exists if and only if the infinite sites and no recombination assumptions hold. In our example, the multifurcating tree displayed in Figure \ref{fig:perfect_phylo}(a) is the corresponding perfect phylogeny representation of the incidence matrix. The key in the perfect phylogeny representation is that mutations (labeled as $s_{1},$\dots$,s_{4}$ in Figure \ref{fig:perfect_phylo}) partition the haplotypes into different groups (3 groups represented as leaf nodes in Figure \ref{fig:perfect_phylo}(a)) and thus enforce a combinatorial constraint.

\begin{figure}[!htbp]
	\includegraphics[clip, width=1\textwidth]{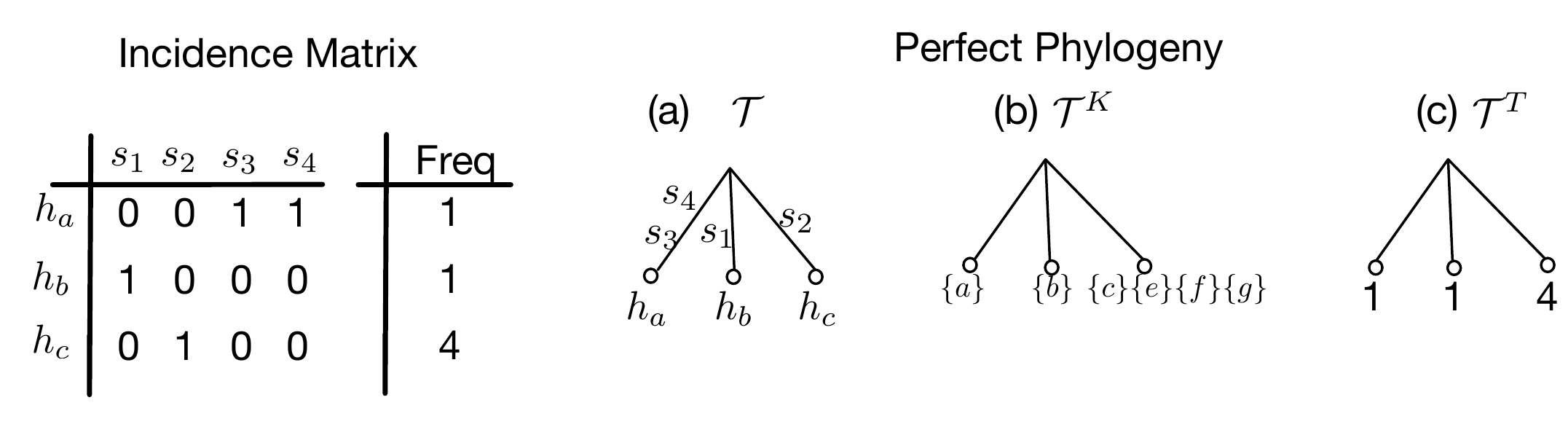}
	\caption{\small{\textbf{Incidence matrix and perfect phylogeny representation}. Data is summarized as an incidence matrix ($h$ denotes the haplotypes, $s$ the segregating sites) and a vector of frequencies. \textbf{(a)} Original perfect phylogeny $\mathcal{T}$ in bijection with the incidence matrix; each of the 4 polymorphic sites labels exactly one edge. When an edge has multiple labels, the order of the labels is irrelevant. Each of the 3 haplotypes labels one leaf if $\mathcal{T}$. \textbf{(b)} Kingman's perfect phylogeny $\mathcal{T}^{K}$:  it is a perfect phylogeny with edge labels removed and leaf labels the set of individual labels for each haplotype. \textbf{(c)} Tajima's perfect phylogeny $\mathcal{T}^{T}$:  it is a perfect phylogeny with edge labels removed and leaf labels the corresponding haplotype frequency.}}
	\label{fig:perfect_phylo}
\end{figure}


More formally, given an incidence matrix $\textbf{Y}$, a \textit{perfect phylogeny} $\mathcal{T}$ is a rooted tree (possibly multifurcating) with $k$ leaves, satisfying the following properties:
\begin{enumerate}
	\item Each of the $k$  haplotypes labels one leaf in $\mathcal{T}$.
	\item Each of the $m$ polymorphic sites labels exactly one edge. When multiple sites label the same edge, the order of the labels along the edge is arbitrary. Some external edges (edges subtending leaves) may not be labeled, indicating that they do not carry additional mutations to their parent node.
	\item For any haplotype $h_{k}$, the labels of the edges along the unique path from the root to the leaf $h_{k}$, specify all the sites where $h_{k}$ has the mutant type.
\end{enumerate}


A few remarks. The tree $\mathcal{T}$ is usually not the tree topology of a coalescent genealogy.· First, each leaf node labels a unique haplotype which could have been sampled with frequency higher than one. Second, we have restricted our attention to binary trees, those sampled from one of the coalescent processes, and $\mathcal{T}$ is not necessarily binary (in most cases it is not).

To simplify our exposition in the following sections, we summarize the perfect phylogeny somewhat different than the original Gusfield's algorithm depending on whether we wish to count Kingman's or Tajima's topologies compatible with the observed data. Our perfect phylogeny representation for counting Kingman's tree topologies is denoted by $\mathcal{T}^{K}$.  In $\mathcal{T}^{K}$, we remove the edge labels and label the leaf nodes by the set of individuals labels  that share the same haplotype. For example in Figure \ref{fig:perfect_phylo}(b) individuals $\{c\},\{d\},\{e\}$ and $\{f\}$ share the same haplotype $h_c$. In the case when a haplotype leaf descends from an edge with no mutations, we attach the set of individuals labels to its parent node and remove the leaf. Similarly, our perfect phylogeny representation for counting Tajima's tree topologies is denoted by $\mathcal{T}^{T}$. In $\mathcal{T}^{T}$, we again remove edge labels but now we label leaf nodes by the frequency of their corresponding haplotypes (Figure \ref{fig:perfect_phylo}(c)). Note that such a representation reflects the fact that two individuals sharing the same mutations are indistinguishable. In the case when a haplotype leaf descends from an edge with no mutations, we attach the frequency of the haplotype to its parent node and remove the leaf.

\begin{figure}[!htbp]
	\includegraphics[clip, width=1\textwidth]{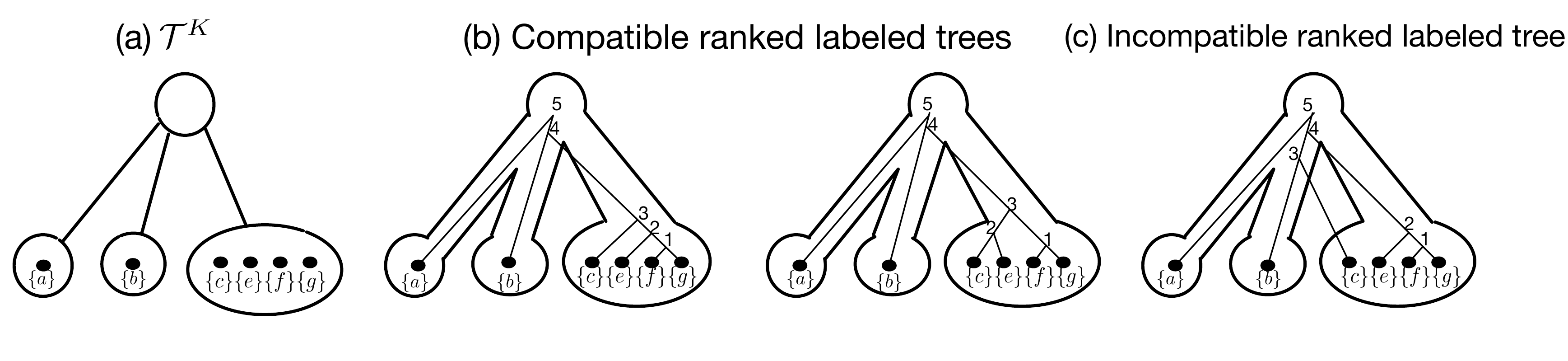}
	\caption{\small{\textbf{Compatibility of ranked labeled trees with the perfect phylogeny.} \textbf{(a)} Kingman perfect phylogeny, \textbf{(b)} two examples of ranked labeled trees compatible with the perfect phylogeny, \textbf{(c)} incompatible ranked labeled tree ($\{c\}$ should coalesce first with individuals in its node).}}
	\label{fig:compatible}
\end{figure}




A tree topology $g$ is \textit{compatible} with the perfect phylogeny $\mathcal{T}$ if $P(\mathcal{T} | g, \textbf{t}) >0$. That is, if all sequences descending from a node $V$ in $\mathcal{T}$ coalesce in $g$ before coalescing with any other sequence descending from a different node $U$ in $\mathcal{T}$. Figure \ref{fig:compatible}(b) shows examples of two compatible ranked labeled trees with the perfect phylogeny in Figure \ref{fig:perfect_phylo}(b) and \ref{fig:compatible}(a), while Figure \ref{fig:compatible}(c) shows an incompatible ranked  labeled tree topology. The topology in Figure \ref{fig:compatible} (c) is not compatible since there is no node in $g^{K}$ that groups together $\{c\},\{e\},\{f\}, \{g\}$ without $\{a\}$ or $\{b\}$.
In the following sections we describe our algorithms for approximating the number of tree topologies compatible with a given perfect phylogeny. 
In the following, we denote the set of compatible tree topologies by $\mathcal{G}_{n,C} \subseteq \mathcal{G}_{n}$.

\section{Sequential importance sampling}
\label{sis_base}

Let $p$ denote the uniform discrete distribution on $\mathcal{G}_{n,C}$. Suppose we can sample from a distribution $q$ with support $\mathcal{G}_{n,C}$, then the normalizing constant of $p$, i.e. $|\mathcal{G}_{n,C}|$ is given by 
\begin{equation}
\label{eq:count}
\E_q \left[\frac{1}{q(g)}\right]= \sum_{g \in \mathcal{G}_{n,C}} \frac{1}{q(g)} q(g) = |\mathcal{G}_{n,C}|,
\end{equation}
which, given an $i.i.d.$ sample from $q$ of size $N$, can be approximated via Monte Carlo by 
\begin{equation}
\label{eq:count_is}
\widehat{|\mathcal{G}_{n,C}|}= \frac{1}{N} \sum_{i=1}^N \frac{1}{q(g_i)},
\end{equation}
with standard error:
$se(\widehat{|\mathcal{G}_{n,C}|})=\sqrt{\Var_q( 1/q(g))}/\sqrt{N}$, and
the variance can be approximated with its empirical counterpart. 

Average \eqref{eq:count_is} is an instance of importance sampling (IS) \citep{ham64,owen13}. 
As described in previous sections, observed data impose combinatorial constraints to the space of compatible tree topologies. The idea is to construct a compatible tree topology $g \in \mathcal{G}_{n,C}$ sequentially with choices $c_{n},\ldots,c_{1}$ (one coalescence at a time) from the tips to the root, ensuring that each choice is compatible with the observed data (or perfect phylogeny) and with known probability:
\begin{equation}
\label{seq_q}
q(g)= q(c_n)  q(c_{n-1} \mid c_{n}) \ldots q(c_{1}  \mid c_{2}),
\end{equation}
Approaches with a similar stochastic sequential nature construction have been used for enumeration  in other contexts, such as random graphs, networks and contingency tables \citep{knu76, chen05,bli11,chen18, dia18}.
It is clear from this literature that the algorithm should satisfy two desiderata: it should not ``get stuck", i.e. it should not sample $g$ outside $|\mathcal{G}_{n,C}|$; in addition, $q(g)$ should be easily computed. 
%
%
How large $N$ should be largely depends on how close the proposal distribution $q$ is to the target distribution $p$. In our problem $p$ is uniform discrete on the set of compatible trees. 
\cite{chat18} show that $N \approx \exp(KL(q,p))$  is necessary and sufficient for accurate estimation by IS, where $KL$ denotes the Kullback-Leibler divergence. In addition, \cite{chat18} warn against the use of sample variance as a criteria for IS convergence: they prove that it can be arbitrary small for large $N$ independently from $p$ and $q$. 

A common metric to assess convergence is the importance sampling effective sample size $\text{ESS}$, where $\text{ESS}=N/(1+ cv^2)$, and $cv^2$ is the coefficient of variation given by
$$cv^2=\frac{\Var_q[p(g)/q(g)]}{\E^{2}_q[p(g)/q(g)]},$$
and estimated empirically. $cv^2$ is the $\chi^2$-distance between $p$ and $q$. A low $cv^2$ (ESS close to $N$), is a good indicator of the quality of the proposal $q$. 


In lieu of sample variance as a metric for convergence, \cite{chat18} define $q_N=\E [Q_N]$ where 
$$Q_N= \frac{\max_{ 1 \leq i \leq N} p(g_i)/q(g_i)}{\sum_{i=1}^N p(g_i)/q(g_i)},$$ 
and propose to use a Monte Carlo estimate of $q_N$ below a certain threshold as a criteria for convergence. A low value of $q_n$ can  be interpreted as a situation in which a sufficiently large number of samples have been collected (large denominator) to counterbalance the effect of possible ``outliers" that are sampled (large numerator). Computing a Monte Carlo estimate is computationally expensive and hence, in this work we simply compute a single running $Q_N$ and combine it with the other metrics described. 
Note that since we restrict our attention to $p$ uniform discrete, the normalizing constant cancels out in both $Q_N$ and $cv^2$, so it is possible to compute these two diagnostics.

\subsection{Sampling trees compatible with a perfect phylogeny}
\label{algo}

To generate a tree topology $g\in\mathcal{G}_{n,C}$ compatible with the observed data $\mathcal{T}$, we proceed sequentially from tips to the root in both $\mathcal{T}$ and $g$: one coalescence in $g$ and one node in $\mathcal{T}$ at a time. In every step we keep an active set of nodes of $\mathcal{T}$ in which we can sample particles to coalesce. Initially, this set includes all nodes with at least two particles. We then randomly select an active node in $\mathcal{T}$ and randomly select two particles from the selected node to coalesce in $g$. At this time, the two selected particles are replaced by a new particle in the selected node in $\mathcal{T}$. If a node in $\mathcal{T}$ has a single particle, the node is removed and its particle is transferred to its parent node. The algorithm ends when $\mathcal{T}$ has a single node with a single particle and when a complete genealogy is generated. All particles in one node must coalesce with each other before they can coalesce with any other particle.

We propose two new algorithms that share the steps just described: one for sampling ranked labeled trees (Kingman trees) and one for sampling ranked tree shapes (Tajima trees). A simple combinatorial argument allows to extend the outputs of these two algorithms to the their respective unranked counterparts. This extension should be considered by all means a byproduct of the Kingman and Tajima algorithms. It is of interest because we can obtain estimates for two other resolutions at almost no additional computational cost.

We start with some notations. We use $V$ to denote the set of nodes of the perfect phylogeny $\mathcal{T}$ and $L \subset V$ to denote the set of active nodes, i.e. nodes with at least two particles; $v$ is an element of $V$, and $\text{pa}(v)$ denotes the parent node of $v$ (if $v$ is not the root). We use the word particle to refer to individual singletons, elements of a partition of $[n]$ or vintages. Each node in $\mathcal{T}$ has either no particles or a given number of particles assigned (labeled or not). 
Given $n$ individuals, the $n-1$ iterations required to sample a tree topology are indexed in reverse order, i.e. from $n-1$ to $1$, to be consistent with the notations used in the jump chains of the $n$-coalescent. This notation allows us to keep track how many individuals have yet to coalesce.

We saw that Kingman $n$-coalescent jump chain induces a uniform distribution with support $\mathcal{G}_n^K$. Given that our target distribution for the Kingman topology is uniform on  $\mathcal{G}_{n,C}^K$, we mimic within each node the transition probability of the underlying coalescent jump chain. The active node is sampled with probability proportional to the number of assigned particles. Although Tajima's jump chain does not induce a uniform distribution on $\mathcal{G}_n^T$, it is quite close to be uniform: it is uniform across ranked tree shapes with the same number of cherries.


\subsubsection{Data constrained Kingman coalescent}
To sample a ranked labeled tree $g^{K}=\{c_{i}\}_{i=n:1}$ of $n$ individuals compatible with the observed perfect phylogeny $\mathcal{T}^{K}$, we start at $c_n=\{\{1\},\ldots,\{n\}\}$. Each leaf node of $\mathcal{T}^K$ defines a partition of $c_n$, and we use $c_n^v$ to denote the set of particles in node $v$. An exception occurs when an edge with no mutations subtends a leaf; in this case the parent node gets assigned the particles of the leaf and the leaf is removed.

The first step is to define the set $L$ as the set of nodes with at least two particles. If a node has a single particle ($|c_n^v|=1$), we transfer the particle to its parent node. Then for each iteration $i=n-1,\ldots,1$, we sample a node in $L$ with probability proportional to the number of particles in that node: at iteration $i$, the probability of choosing node $v_i \in L$ is  $q(v_i)=|c_{i+1}^{v_i}|/\sum_{j \in L} |c_{i+1}^j|$.
The transition from $c_{i+1}^{v_i}$ to $c_{i}^{v_{i}}$ consists in joining two particles of $c_{i}^{v_{i}}$ uniformly at random. If a node is not sampled, we assume $c_{i}^{v}=c_{i+1}^{v}$. This choice mimics the jump chain of a Kingman $n$-coalescent; the difference is that the Markov chain moves one step on a constrained state space: $c_i^{v_i}$ in lieu of $c_i$; \textit{i.e.} the coalescent event in node $v_i$ has probability:
\begin{equation}
\label{king_algo_move}
q(c^{v_i}_{i} \mid c^{v_i}_{i+1})=\left\{
\begin{array}{ll}
\dbinom{|c^{v_i}_{i+1}|}{2}^{-1} \,\ \,\ \,\ \,\ \text{if}  \ \,\ c^{v_i}_{i} \prec c^{v_i}_{i+1}\\
\,\ \\
0 \,\ \,\ \,\ \,\  \,\ \,\ \,  \,\ \,\  \text{otherwise}
\end{array}.
\right.
\end{equation}

\noindent Note that at every iteration $c_i =\cup_{v} c_i^v$. The two probabilities $q(v_i)$ and $q(c^{v_i}_{i} \mid c^{v_i}_{i+1})$ are all we need to compute the transition probability  
$$q(c_{i} \mid c_{i+1})=q(v_i) q(c^{v_i}_{i} \mid c^{v_i}_{i+1}),$$
where $c_i= c_{i+1} \setminus c_{i+1}^{v} \cup c_i^{v}$ can be constructed recursively. The last iteration happens at the root node of $\mathcal{T}^K$ and  $q(g^T)$ is computed as the product of the transition probabilities as in \eqref{seq_q}. 
We outline our sampling algorithm with the following example and provide the pseudocode in Algorithm \ref{alg:algo1}.

\begin{example} Consider the perfect phylogeny $\mathcal{T}^K$ in Figure \ref{fig:king_sample}(a). To avoid confusion between the nodes' sampling order $(v_{n-1}, \ldots, v_{1})$ and node labels, we label the root node $j_0$ and the leaf nodes $j_1$, $j_2$ and $j_3$. Figure \ref{fig:king_sample} gives a graphical representation of a single run of the algorithm, where one particle is assigned to $j_1$, one to $j_2$ and four to $j_3$. We start with 
	$c_6^{j_1}= \{a\}$, $c_6^{j_2}= \{b\}$,  $c_6^{j_4}= \{\{c\},\{d\},\{e\},\{f\}\}$ and $c_6^{j_0}=\emptyset$. Now, both $j_1$ and $j_2$ have a single particle: we transfer their particles to the root node and update $c_n^{j_0}=\{\{a\},\{b\}\}$ (Figure \ref{fig:king_sample}(a-b)). The set of active nodes is $L=\{j_0,j_3\}$.  
	At iteration $i=5$ (first iteration) suppose we sample node $v_5=j_3$, this happens with probability $4/6$. then $d$ and $f$coalesce with probability $1/6$ (Figure \ref{fig:king_sample}(b)). We update $c_{5}^{j_3}=\{\{c\},\{e\},\{d, f\}\}$. The set of active sample nodes remains $L=\{j_0,j_3\}$.  Figure \ref{fig:king_sample}(c-f) shows the remaining iterations.  The sequence of sampled nodes is $\{v_5=j_3,v_4=j_3,v_3=j_0,v_2=j_3,v_1=j_0\}$ with sampling probabilities $(4/6,3/5,1/2,1,1)$. The coalescent events probabilities  are $(1/6,1/3,1,1,1)$. Thus, $q(g^K)=1/90$. 
	\begin{figure}[!htbp]
		\includegraphics[clip, width=1\textwidth]{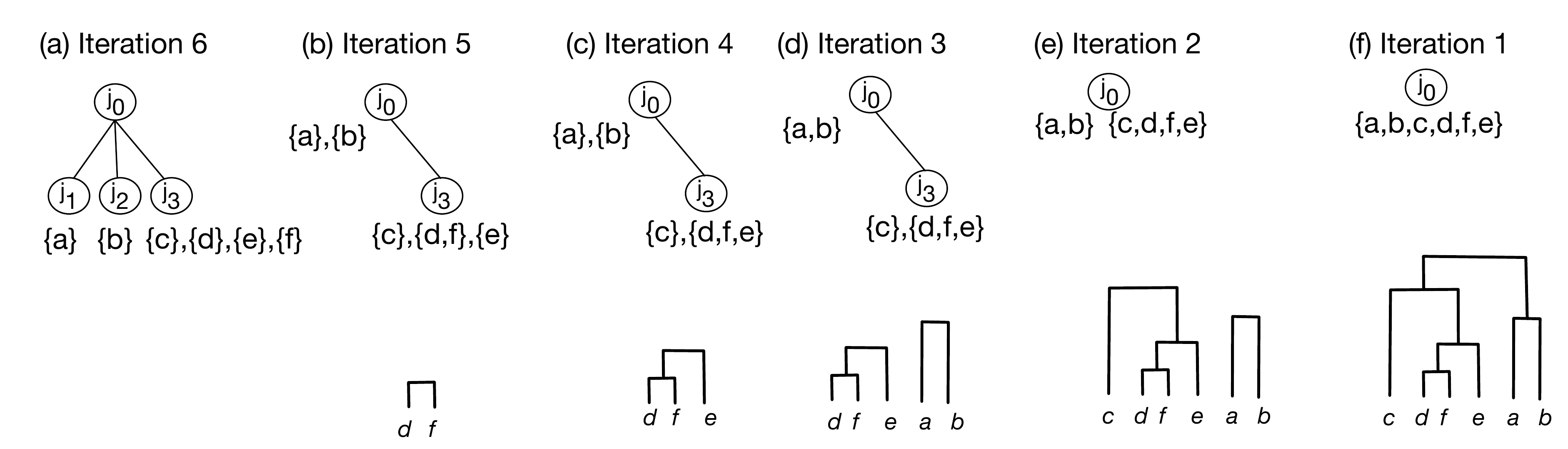}
		\caption{\textbf{Example 1: sequential sampling of a Kingman tree topology with constraints.} First row describes the steps in the perfect phylogeny; second row describes how $g^K$ is sequentially sampled. We start with a perfect phylogeny (a), in (b) we assign the singletons to their parent node. In this case $\{a\}$ and $\{b\}$ are assigned to $j_0$. At each iteration (b)-(f) we select a node and coalesce a pair from the selected node. The algorithm terminates when a single tree topology of size $n$ is generated.}
		\label{fig:king_sample}
	\end{figure}
	
\end{example}
%
\begin{algorithm} [h]
	\caption{Sequential sampling on a constrained Kingman tree topology}
	\label{alg:algo1}
	\begin{algorithmic}
		\State \textbf{Inputs:} $\mathcal{T}^K$ with $c^{v}_{n}$ subsets of singletons at all nodes with particles and $c^{v}_{n}=\emptyset$ at all remaining nodes.
		\State \textbf{Outputs:} $g^K$, $q(g^K)$
		\begin{enumerate}
			\item If a node $v$ is such that $|c_n^v|=1$, then we let $c_n^{\text{pa}(v)}=c_n^{\text{pa}(v)} \cup c_n^v$ and $c_n^v=\emptyset$.
			\item Define $L$ as the list of nodes such that  $|c_n^v|>1$
			\item Initialize $q=1$
			\For {$i=n-1$ to $1$}
			\begin{enumerate}
				\item Sample node $v_i$ in $L$ with probability $q(v_i)$.  
				\item Choose particles in $v_{i}$ to coalesce with probability $q(c_{i} \mid c_{i+1})$.
				\item Update $c^{v_{i}}_{i}$ and define $c^{v}_{i}=c^{v}_{i+1}$ for all the other nodes.
				\item If  $|c^{v_{i}}_{i}|=1$, we let $c_{i}^{\text{pa}(v_{i})}=c_{i}^{\text{pa}(v_{i})} \cup c_{i}^{v_{i}}$ and   $c^{v_{i}}_{i-1}=\emptyset$.
				\item Update $q=q\times q(v_i) \times q(c_{i} \mid c_{i+1})$ 
				\item Update $L$ as the list of nodes such that  $|c_{i}^v|>1$ 
			\end{enumerate}
			\EndFor
		\end{enumerate}
	\end{algorithmic}
\end{algorithm}
\subsubsection{Data constrained Tajima coalescent}

To sample a ranked tree shape $g^{T}=\{(\alpha_{i},\beta_{i})\}_{i=n:1}$ of $n$ individuals compatible with the observed perfect phylogeny $\mathcal{T}^T$ (Figure \ref{fig:perfect_phylo} (c)), we start at $(n,\emptyset)$ and each leaf node in the perfect phylogeny $\mathcal{T}^T$ is assigned a vector  $(\alpha^v_{n},\beta^v_{n})$. Again, an exception occurs when an edge with no mutations subtends a leaf. In this case, we assign the particles to their parent node and remove the leaf node. Recall that $\alpha^{v}_{n}$ denotes the number of singletons, and $\beta^v_{n}$ denotes the set of vintages associated to node $v$. Initially, each leaf node (possibly also some internal nodes when the no-mutations case occurs) in the perfect phylogeny contains the number of singleton particles $ \sum_{v \in V} \alpha^v_{n}=n$, and no vintages, i.e. $\beta^v_n= \emptyset$ for all $v \in V$. At any given iteration $i$, the number of particles associated to a node $v$ is  $\alpha^{v}_{i}+|\beta^v_{i}|$.

Tajima's sampler follows the rationale used to build the Kingman sampler. We define the set $L$ as in the Kingman's sampler (nodes with at least two particles). Then for $n-1$ iterations, we first sample a node $v \in L$ with probability $q(v_i)=(\alpha^{v}_{i}+|\beta^v_{i}|) /\sum_{j \in L}(\alpha^{j}_{i}+|\beta^j_{i}|)$; then we sample a pair of particles in the selected node to coalesce. Our proposal  probability is 
is: 

\begin{equation}
\label{eq:taj_algo1}
q\Big[(\alpha^{v_i}_{i},\beta^{v_i}_{i})\Big| (\alpha^{v_i}_{i+1},\beta^{v_i}_{i+1})\Big]=\left\{
\begin{array}{ll}
\dbinom{\alpha^{v_i}_{i+1}}{\alpha^{v_i}_{i+1}-\alpha^{v_i}_{i}}\dbinom{\alpha^{v_i}_{i+1}+|\beta^{v_i}_{i+1}|}{2}^{-1} \,\ \,\ \,\ \,\ \text{if} \ \,\  (\alpha^{v_i}_{i},\beta^{v_i}_{i}) \prec  (\alpha^{v_i}_{i+1},\beta^{v_i}_{i+1})\\
\,\ \\
0 \,\ \,\ \,\ \,\  \,\ \,\ \,  \,\ \,\  \text{otherwise}
\end{array}
\right.
\end{equation}

Analogously to the Kingman sampler, each iteration ends by updating $(\alpha^{v_i}_{i},\beta^{v_i}_{i})$ and $L$. The pseudocode is presented in Algorithm 2.  Note that, as opposed to the Kingman sampler,  $q(v_i)$ and $q\Big[(\alpha^{v_i}_{i},\beta^{v_i}_{i})\Big| (\alpha^{v_i}_{i+1},\beta^{v_i}_{i+1})\Big]$ in Tajima sampling, do not fully determine $q\Big[(\alpha_{i},\beta_{i})\Big| (\alpha_{i+1},\beta_{i+1})\Big]$, where $(\alpha_{i},\beta_{i})$ is the $i$th state independent of which node is selected, it can be computed as $(\alpha_{i},\beta_{i})=(\sum_{v \in V} \alpha^v_{i}, \bigcup_{v \in V} \beta^v_{i})$. A transition from $(\alpha_{i+1},\beta_{i+1})$ to $(\alpha_{i},\beta_{i})$ can be obtained by sampling different nodes in the active set, possibly with different sampling probabilities.  For example, suppose we are joining two singletons: any $v \in L$ with at least two singletons allows this type of transition. 
This issue was not relevant in the Kingman sampler because individuals were labeled. 
Therefore, the output of the sampling algorithm after $n-1$ iterations is $\{(\alpha_{i},\beta_{i})\}_{i=n:1}=g^T$ along with the sequence of sampling nodes $\bm{v}=(v_{n-1}, \ldots, v_1)$. 
It is possible to sample the same $g^T$ with different $\bm{v}$ and $\bm{v'}$. 
These two outputs of the algorithm, which we denote by $(g^T, \bm{v})$ and $(g^T, \bm{v'})$, may also have different sampling probabilities $q(g^T, \bm{v})$ and $q(g^T, \bm{v'})$. We illustrate this situation with the following example.
\begin{figure}[!htbp]
\center{\includegraphics[width=0.7\textwidth]{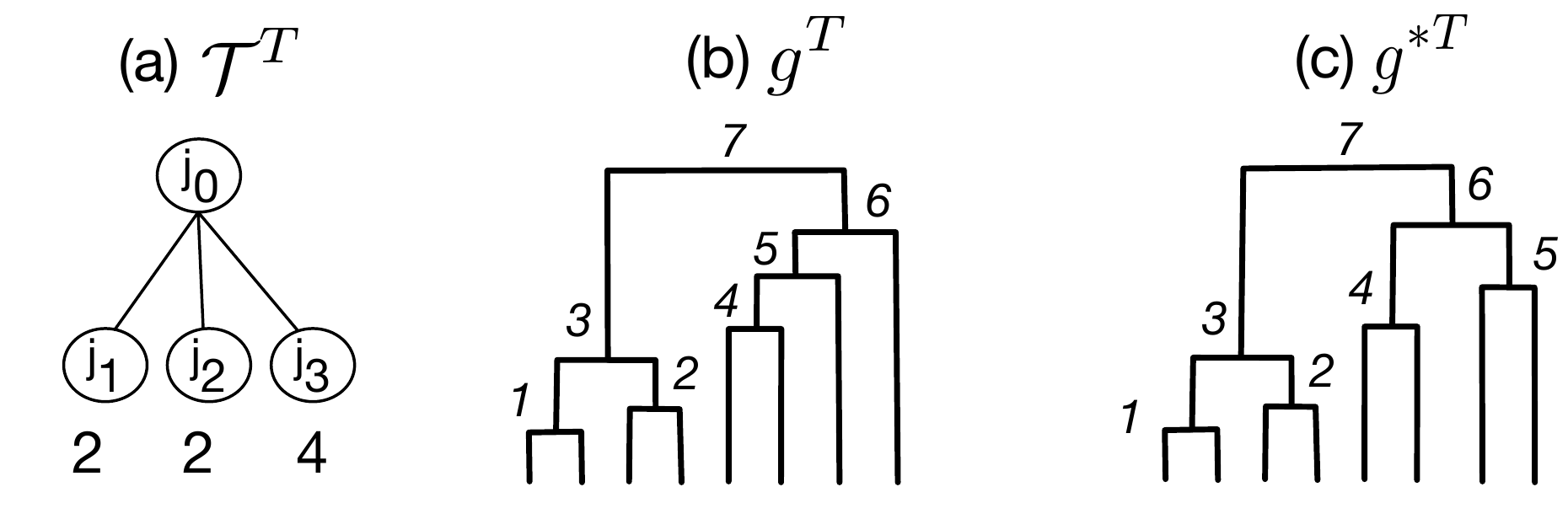}}
	\caption{\textbf{Example 2: two ranked tree shapes compatible with a given perfect phylogeny.}  (a) perfect phylogeny (b)-(c) two possible ranked tree shapes compatible with $\mathcal{T}^T$. Tree (b) can be sampled through two node orderings $\bm{v}=\{ j_1,j_2,j_0,j_3,j_3,j_3,j_0 \}$ and $\bm{v'}=\{ j_2,j_1,j_0,j_3,j_3,j_3,j_0 \}$, tree (c) through four orderings $\bm{v}$, $\bm{v'}$, $\bm{v''}=\{ j_3,j_3,j_3,j_1,j_2,j_0,j_0 \}$ and $\bm{v'''}=\{ j_3,j_3,j_3,j_2,j_1,j_0,j_0 \}$.}
	\label{fig:taj_sample}
\end{figure}

\begin{example}
	Consider the perfect phylogeny in Figure \ref{fig:taj_sample} (a).  Figure \ref{fig:taj_sample} (b)-(c) show two ranked tree shapes, $g^T$ and $g^{*T}$, that can be sampled with our algorithm. Let us first consider $g^T$ in Figure \ref{fig:taj_sample}(b). A possible sequence of sampling nodes in $\mathcal{T}^{T}$ is $\bm{v}=\{ j_1,j_2,j_0,j_3,j_3,j_3,j_0 \}$. In this case the output of Algorithm 2 would be $(g^T,\bm{v})$. Although, the sequence $\bm{v'}=\{ j_2,j_1,j_0,j_3,j_3,j_3,j_0 \}$ leads also to the same $g^T$. The two node orderings $\bm{v}$ and  $\bm{v'}$ can be easily identified in $\mathcal{T}^T$ since nodes $j_1$ and $j_2$ are indistinguishable by being siblings of the same size. 
	Let us now turn to $g^{*T}$ in Figure \ref{fig:taj_sample}(c). In this case, there are 4 possible sampling nodes orderings:  $\bm{v}$, $\bm{v'}$, $\bm{v''}=\{ j_3,j_3,j_3,j_1,j_2,j_0,j_0 \}$ and $\bm{v'''}=\{ j_3,j_3,j_3,j_2,j_1,j_0,j_0 \}$.
\end{example}

We now introduce some notation to distinguish between the output of our sampling algorithm and the elements needed in the sequential importance sampling estimation of $|\mathcal{G}^{T}_{n,c}|$.
\begin{definition}
	Let  $\mathcal{Y}_{n,C}^T$ be the set of all possible outcomes $(g^T, \bm{v})$ of the Tajima algorithm (Algorithm 2) conditionally on a given perfect phylogeny $\mathcal{T}^{T}.$
	We call two outputs of the algorithm: $(g^T, \bm{v})$ and $(g^T, \bm{v'})$ equivalent if they have the same ranked tree shape $g^T$. Define $c^T(g^T)$  be the size of the equivalence class, that is the number of possible pairs $(g^T, \bm{v'}) \in \mathcal{Y}_{n,C}^T$ equivalent to $(g^T, \bm{v}) $. 
\end{definition}



It is still possible to use sequential importance sampling despite the fact that our proposal $q$ has support $\mathcal{Y}_{n,C}^T$ instead of $\mathcal{G}_{n,C}^T$. We discuss two alternative ways. The first one is to generate a sample $(g^{T},\bm{v}) \in \mathcal{Y}_{n,C}^T$ with sampling probability  $q(g^{T},\bm{v})$ computed as the product of all  transition probabilities (Algorithm 2). We then call a backtracking algorithm that lists all possible sequence of nodes $\bm{v'}$ that would give rise to the same $g^{T}$ and compute: 
\begin{equation}
\label{equi_class}
q(g^T)=\sum_{\bm{v'}:(g^T,\bm{v'}) \in \mathcal{Y}_{n,C}^T} q(g^T,\bm{v'}).
\end{equation}
Finally, we estimate the cardinality of our constrained space by the Monte Carlo approximation to the following:
\begin{align}
\label{eq:tajima_sis}
E_{\mathcal{Y}^{T}_{n,C}}\left[\frac{1}{q(g^{T})}\right]&=  \sum_{(g^{T},\bm{v})\in\mathcal{Y}^{T}_{n,C} } \frac{q(g^{T},\bm{v})}{q(g^{T})} =\sum_{g^{T}\in \mathcal{G}^{T}_{n,C}}\frac{1}{q(g^{T})} \sum_{\bm{v}:(g^{T},\bm{v})\in\mathcal{Y}^{T}_{n,C} }q(g^{T},\bm{v}) \nonumber \\
&=\sum_{g^{T}\in \mathcal{G}^{T}_{n,C}}\frac{q(g^{T})}{q(g^{T})}\nonumber =|\mathcal{G}^{T}_{n,C}| \nonumber\\
\end{align}

Note that the backtracking algorithm adds a computational burden to the procedure.  The complexity cannot be uniquely determined and, as it is known in the backtracking literature, may vary largely from problem to problem \citep{knu18}. Since the complexity depends both on the data $\mathcal{T}^T$ and $g^T$, an analytical expression is not available. We will study this computational burden through simulations in Section \ref{simu}.

An alternative to the backtracking step is desirable but currently still an open problem. A potential alternative is inspired by a similar situation discussed in \cite{bli11} in the context of sampling graphs with a given degree sequence. The cardinality is estimated by the Monte Carlo approximation to the following \begin{align}
\label{eq:tajima_sis2}
E_{\mathcal{Y}^{T}_{n,C}}\left[\frac{1}{c^{T}(g^{T})q(g^{T},\bm{v})}\right]&=  \sum_{(g^{T},\bm{v})\in\mathcal{Y}^{T}_{n,C} } \frac{q(g^{T},\bm{v})}{c^{T}(g^{T})q(g^{T},\bm{v})} \nonumber \\
&=\sum_{g^{T}\in \mathcal{G}^{T}_{n,C}}\frac{1}{c^{T}(g^{T})} \sum_{\bm{v}:(g^{T},\bm{v})\in\mathcal{Y}^{T}_{n,C} }1 =|\mathcal{G}^{T}_{n,C}| ,\nonumber\\
\end{align}
where $c^T(g^T)$ is the size of the equivalence class $
c^{T}(g^T)=\#\{\bm{v'}:(g^T,\bm{v'}) \in \mathcal{Y}_{n,C}^T\}$ as in Definition 3. Given a pair $(g^{T},\bm{v})$, we can calculate $c^T(g^T)$ by finding equivalence classes of certain subtrees in $g^{T}$ relative to $\mathcal{T}^{T}$. Although  a practical implementation is computationally prohibitive, we introduce this idea because in the next section we use it to obtain unranked tree topologies (labeled trees and tree shapes) algorithms as a byproduct of the Kingman and Tajima algorithms.

\begin{algorithm} [h]
	\caption{Sampling  on the constrained Tajima Space}
	\begin{algorithmic}\label{algo:algo2}
		\State \textbf{Inputs:} $\mathcal{T}^T$, with $\alpha_{n}^v$ number of singletons at all leaf nodes (plus the parent node if the leaf subtends from a branch with no mutations), and $\beta_{n}^v=\emptyset$ for all $v \in V$.
		\State \textbf{Outputs:} $g^T$, $q(g^T)$
		\begin{enumerate}
			\item If a node $v$ is such that $\alpha_{n}^v=1$, then let $\alpha_{n}^{pa(v)}=\alpha_{n}^{pa(v)}+1$, $\alpha_{n}^v=0$.
			\item Define $L$ as the list of nodes with  $\alpha_{n}^v>1$.
			\For {$i=n-1$ to $1$}
			\begin{enumerate}
				\item Sample node $v_i$ with probability $q(v_i)$.  
				\item Choose particles to coalesce with probability $q\Big[(\alpha^{v_i}_{i},\beta^{v_i}_{i})\Big| (\alpha^{v_i}_{i+1},\beta^{v_i}_{i+1})\Big]$
				\item Update $(\alpha^{v_i}_{i},\beta^{v_i}_{i})$ and define $(\alpha^{v_i}_{i},\beta^{v_i}_{i})=(\alpha^{v_i}_{i+1},\beta^{v_i}_{i+1})$ for all other nodes
				\item If $\alpha^{v_{i}}_{i}+|\beta^{v_{i}}_{i}|=1$, then let $\alpha_{i}^{pa(v_{i})}=\alpha_{i}^{pa(v_{i})}+\alpha_{i}^{v_{i}}$, $\alpha_{i}^{v_{i}}=0$, and $\beta_{i}^{pa(v_{i})}=\beta_{i}^{pa(v_{i})}\cup\beta_{i}^{v_{i}}$, $\beta_{i}^{v_{i}}=\emptyset$
				\item Update $q=q\times q(v_i)\times q\Big[(\alpha^{v_i}_{i},\beta^{v_i}_{i})\Big| (\alpha^{v_i}_{i+1},\beta^{v_i}_{i+1})\Big]$
				\item Update $L$ as the list of nodes such that $\alpha_i^v+ |\beta^v_i|>1$.
			\end{enumerate}
			\EndFor
			\item  Compute all possible node paths $\textbf{v}$ that lead to $g^T$ (backtracking algorithm).
			\item Compute  $q(g^T)$ as in \eqref{equi_class}. 
		\end{enumerate}
		
	\end{algorithmic}
\end{algorithm}

\subsubsection{Labeled trees and tree shapes}
\label{LTTS}
We now turn to the unranked versions:  
labeled trees and tree shapes.
As before, we define equivalence relations that partitions the spaces $\mathcal{G}_{n,C}^K$ and $\mathcal{G}_{n,C}^T$ into equivalence classes that ignore rankings. We show two simple formulas to compute the size of these classes. As opposed to ranked tree shapes, these formulas are easy to implement and allow to build a SIS procedure to estimate $|\mathcal{G}_{n,C}^{LT}|$ and $|\mathcal{G}_{n,C}^{TS}|$ using outputs from the Kingman and Tajima algorithms (Algorithm 1 and 2). First we define the following two equivalence relations and their cardinalities. 

\begin{definition}
	For any element $g^K \in \mathcal{G}_{n,C}^K$, let $LT(g^K)$ denote the corresponding unranked labeled tree $g^{LT} \in \mathcal{G}_{n,C}^{LT}$, obtained by removing the rankings from internal nodes of $g^K$. We call $g^K$ and $g'^K$ equivalent if $LT (g^K)=LT(g'^K)$ and we denote the size of the equivalence class by $c^{LT}(g^K)$.
\end{definition}

\begin{proposition} \label{gk}
	Let $g^K \in \mathcal{G}_{n,C}^K$, and let $g^K_{i,1}$ and $g^K_{i,2}$ be the two subtrees (or clades) that merge at the $i$th  coalescent event for $i=1,\ldots,n-1$. Then
	$$c^{LT}(g^K)=\prod_{i=1}^{n-1} \frac{(|g^K_{i,1}|+ |g^K_{i,2}|-2)!}{(|g^K_{i,1}|-1)! (|g^K_{i,2}|-1)!},$$
	where $|g^K_{i,j}|$ denotes the number of leaf nodes of $g^K_{i,j}$.
\end{proposition}
\begin{proof} Note that $|g^K_{i,j}|-1$ is the number of coalescent events in subtree $g^K_{i,j}$. For each fixed $i$, we are computing the number of possible permutations of $(|g^K_{i,1}|+ |g^K_{i,2}|-2)$ coalescent events of elements of two groups with  $|g^K_{i,1}|-1$ and $|g^K_{i,2}|-1$ elements respectively. The product accounts for all possible orderings. 
\end{proof}

\begin{definition}
	For any element $g^T \in \mathcal{G}_{n,C}^T$, let $TS(g^T)$ denote the corresponding (unranked) tree shape $g^{TS} \in \mathcal{G}_{n,C}^{TS}$, obtained by removing the rankings from internal nodes of $g^T$. We call $g^T$ and $g'^T$ equivalent if $TS (g^T)=TS(g'^T)$ and we denote the size of the equivalence class by $c^{TS}(g^T)$.
\end{definition}

\begin{proposition} \label{gT}
	Let $g^T \in \mathcal{G}_{n,C}^T$, and let $g^T_{i,1}$ and $g^T_{i,2}$ be the two subtrees (or clades) that merge at the $i$th coalescent event, then
	$$c^{TS}(g^T)=\prod_{i =1 }^{n-1} \frac{(|g^T_{i,1}|+ |g^T_{i,2}|-2)!}{(|g^T_{i,1}|-1)! (|g^T_{i,2}|-1)!}\left(\frac{1}{2}\right)^{1\{g^{T}_{i,1}=g^{T}_{i,2}\}},$$
	where $|g^T_{i,j}|$ denotes the number of leaf nodes of $g^T_{i,j}$.
\end{proposition}

\begin{proof}
	Again, the formula is a product of permutations with repetitions. If the two subtrees that merge at the $i$th coalescence are equal, we need to divided by two since the same rankings in the two subtrees are indistinguishable. 
\end{proof}

Given $c^{LT}(g^K)$ and $c^{TS}(g^T)$, we can easily compute $q(g^{LT})=c^{LT}(g^K) q(g^K)$ and $q(g^{TS})=c^{TS}(g^T) q(g^T)$. These two distributions constitute our sampling proposal  in SIS procedure to estimate $|\mathcal{G}_{n,C}^{LT}|$ and $|\mathcal{G}_{n,C}^{TS}|$.

\section{Simulations} \label{simu} We rely on simulations to assess the convergence, empirical accuracy and computational performance of the proposed algorithms. We discuss a range of scenarios designed to capture a variety of settings encountered in applications and to highlight the key properties of the algorithms. All the algorithms are implemented in the R package \texttt{phylodyn}, which is available for download at \texttt{https://github.com/JuliaPalacios/phylodyn.} The four tree topologies analyzed are $\mathcal{G}_{n,C}^K$: Kingman ranked labeled trees, $\mathcal{G}_{n,C}^T$: Tajima ranked tree shapes,  $\mathcal{G}_{n,C}^{TS}$: tree shapes and $\mathcal{G}_{n,C}^{LT}$: unranked labeled trees. All of which are compatible with the simulated dataset.

We encode our simulated molecular data as an $n \times m$ incidence matrix $\textbf{Y}$ with $n$ sequences and $m$ polymorphic sites. $\textbf{Y}$ is sampled in three steps: we first simulate a Kingman genealogy of $n$ individuals; then the number of mutations $m$ is Poisson distributed, that is
$$(g^K,\textbf{t})\sim \text{Kingman \textit{n}-coalescent},$$
$$M \sim \text{Poisson($\mu$L)}, \,\ \,\ \,\ L=\sum^{n}_{k=2}kt_{k},$$
with $\mu$ denoting the mutation parameter, $L$ the tree length, and \textbf{t} is the $(n-1)$-vector of coalescent times. Coalescent times are exponentially distributed with rate $\binom{k}{2}$ (assuming constant population size). An open source implementation to simulate Kingman genealogies in  \texttt{R} is  \texttt{ape:rcoal()}\citep{par04}. The $m$ mutations are then placed uniformly at random along the branches of the timed genealogy $(g^K,\textbf{t})$ and labeled $1,\ldots,m$. The matrix $\textbf{Y}$ is constructed by setting the $(i,j)$th entry equal to 1 if the branch path from leaf $i$ to the root has labeled mutation $j$. This part of the simulation algorithm corresponds to the infinite-sites mutation model. Finally, $\textbf{Y}$ is summarized by its unique haplotypes (rows) and a vector of haplotype frequencies (incidence matrix in Figure \ref{fig:perfect_phylo}). The corresponding perfect phylogeny $\mathcal{T}$ (Figure \ref{fig:perfect_phylo} (a)) is constructed via \cite{gus91} algorithm;  the Kingman's perfect phylogeny $\mathcal{T}^K$ and the Tajima's perfect phylogeny $\mathcal{T}^T$ are constructed from $\mathcal{T}$ as described in Section \ref{perphylo}.

To assess convergence of our algorithms at various sample sizes and with different combinatorial constraints (defined by the pattern of mutations), we simulate incidence matrices under four scenarios: with sample sizes $n \in (10,20)$, and two mutation regimes $\mu \in (5, 20)$. We computed SIS estimates and diagnostics after $N$ number of iterations with $N \in  (100, 500,  1000, 3000, 5000, 10000,15000)$ from 20 repetitions on each of the four incidence matrices. 

Figure \ref{conv} shows the estimated cardinalities (grey lines) of the four topological spaces (rows) and for the four combinations of $n$  and $\mu$ (columns). Black lines depict the mean estimates. Figure \ref{se/count} plots the ratio of the standard error to the estimated counts (relative SE, rSE) averaged over the 20 SIS runs for each coalescent resolution (distinct colors) and the four simulation scenarios (distinct panels).  The relevant information in Figure \ref{se/count} is not the decay of the lines, which is expected, but rather the order of magnitude of the rSE values when comparing the values across the four algorithms (a lower value means a better empirical performance). Table 1 reports the $cv^2$ values for the four algorithms (rows) and the four combinations of $\mu$ and $n$ (columns). 

\begin{figure}[!htbp]
	\includegraphics[width=\textwidth]{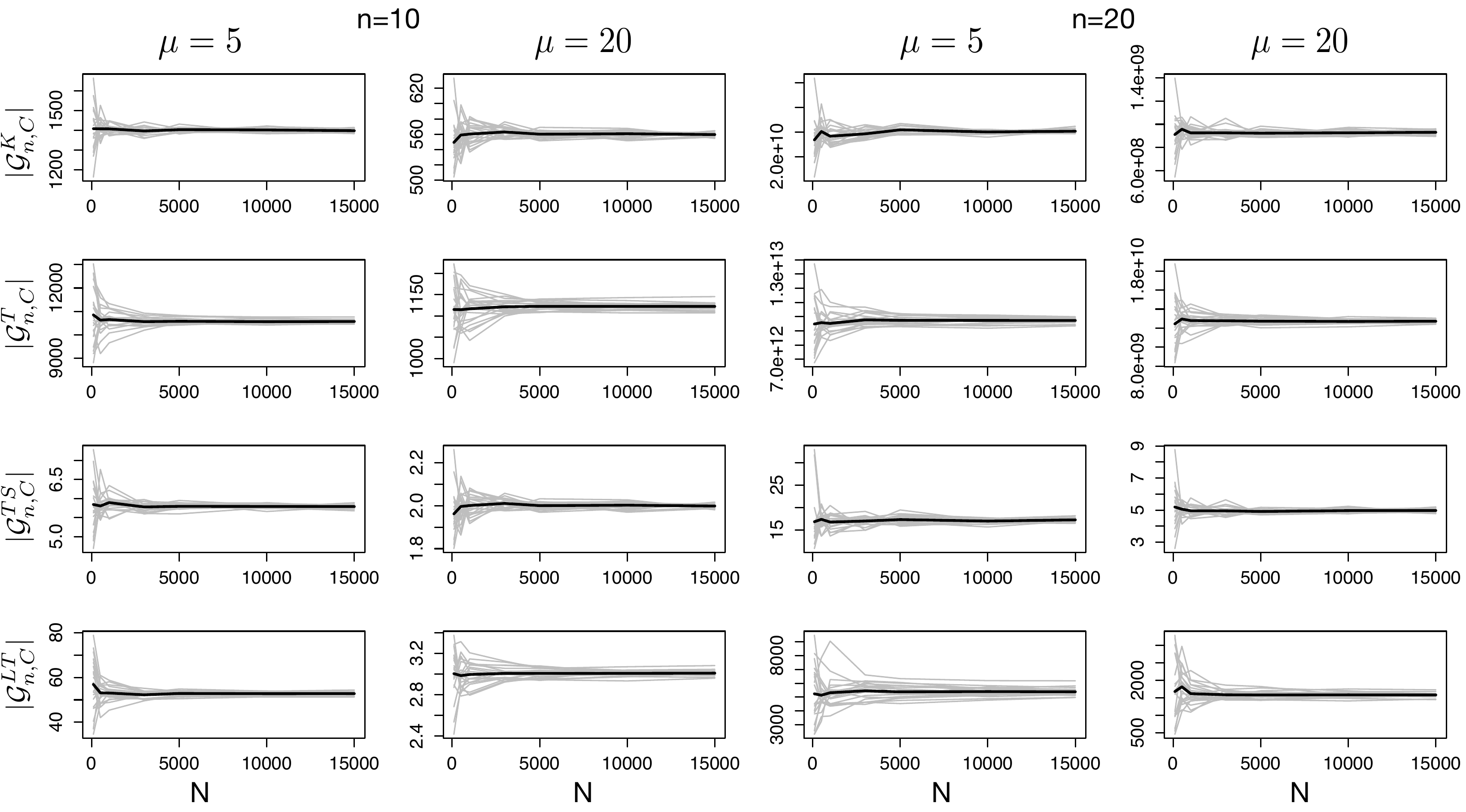}
	\caption{\small{\textbf{Simulations: sequential importance sampling count estimates.} Rows show the estimated cardinality of the four tree topologies: ranked tree shapes ($|\mathcal{G}_{n,C}^{T}|$), ranked labeled trees ($|\mathcal{G}_{n,C}^{K}|$), unranked tree shapes ($|\mathcal{G}_{n,C}^{TS}|$) and labeled trees ($|\mathcal{G}_{n,C}^{LT}|$) (top to bottom rows), the first two columns correspond to simulations based on $n=10$ samples and the last two columns on $n=20$ samples. The first and third columns correspond to mutation rate $\mu=5$ and second and fourth to $\mu=20$. Grey lines correspond to each of the 20 independent estimates from the 20 repetitions of the SIS algorithm computed at $N\in  (100, 500, 1000, 3000, 5000, 10000,15000)$ iterations. Black lines show the mean estimate of the 20 repetitions.}}
	\label{conv}
\end{figure}

\begin{figure}[!htbp]
	\includegraphics[width=\textwidth]{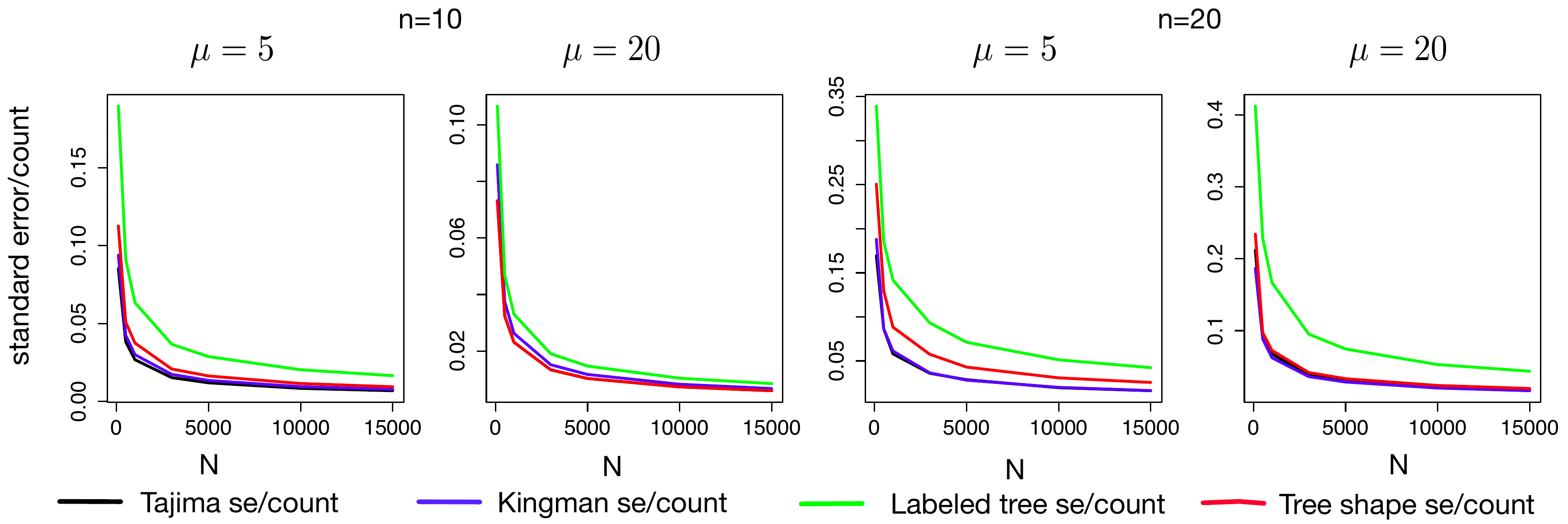}
	\caption{\small{\textbf{Simulations: ratio of standard error to approximate count.} The lines correspond to the four topologies: ranked tree shapes (black), ranked labeled trees (blue), unranked tree shapes (red) and labeled trees (green) respectively. The first two columns correspond to simulations based on $n=10$ samples and the last two columns on $n=20$ samples. The first and third columns correspond to mutation rate $\mu=5$ and second and fourth correspond to  $\mu=20$. Lines plot the ratio of the standard error to the average estimated count  over the 20 repetitions of the SIS algorithm computed at $N\in  (100, 500, 1000, 3000, 5000, 10000,15000)$ iterations. Note that black and blue lines practically overlap.}}
	\label{se/count}
\end{figure}

Figure \ref{conv} provides a visual intuition on the number of MC samples required for convergence.  The four algorithms estimates stabilize around the mean in the four simulation scenarios after $N=5000$. The variable degree at which the grey lines are relatively scattered around the mean, hints that the empirical convergence rates deteriorate for larger sample sizes and unranked topologies. Both rSE (Figure \ref{se/count}) and $cv^2$ (Table \ref{ex1cv2})  confirm and quantify this intuition. In Figure \ref{se/count}, we observe an increase of rSE  with sample size (two to four times higher when increasing $n$ from $10$ to $20$).  Similarly, the rSEs of unranked counts (green and red lines) are one to four times higher than the ranked counterparts (black and blue lines). Similarly, the $cv^2$ values in Table \ref{ex1cv2} increase for larger sample sizes (four to twenty times higher) and for unranked algorithms (one to eight times higher than the $cv^2$ values for ranked algorithms).

These two trends are expected and have a clear explanation. The effect of a sample size increase is twofold: first, the state space (not conditioning on the data) becomes much larger; second, the expected number of mutations increases. A higher number of mutations is likely to impose more constraints in the spaces of genealogies and consequently undermine the algorithm performance (we will elaborate on this point below). The poor performance of both unranked algorithms was also expected since our importance sampling proposals for unranked spaces are not ``tailored" to the underlying coalescent resolutions but obtained by correcting the ranked estimates through the equivalence classes coefficients in Proposition \ref{gk} and  \ref{gT} (Section \ref{LTTS}).

A second rather surprising result is the superior performance of the tree shape algorithm (red line in Figure \ref{se/count} and third row  Table \ref{ex1cv2}) when compared to the labeled tree algorithm (green line in Figure \ref{se/count} and last row  Table \ref{ex1cv2}). This was partially unexpected; it may be a consequence of the fact that the equivalence classes in the unranked spaces are, for most cases, smaller than in the labeled spaces, i.e. $c^{TS}(g^T) < c^{LT}(g^K)$. In contrast to this observation,  we note that neither $cv^2$ (Table \ref{ex1cv2}) nor rSE (black and blue lines in Figure \ref{se/count}) provide evidence of different performance between the two ranked algorithms (black and blue lines  in Figure \ref{se/count} overlap). This is also surprising since we expect our proposal for the Kingman algorithm to be ``closer" to uniform than the Tajima algorithm. We will further investigate this point.

Table \ref{ex1cv2} highlights a very poor performance of the unranked methods for $n=20$, especially for labeled trees. A $cv^2$ close to $20$ raises serious concerns on the reliability of the unranked algorithms in this context, suggesting the risk of a variance explosion ($cv^2$ is a rescaled variance) and low efficiency (ESS is about $5\%$ of the chosen $N$). Similar $cv^2$ has been achieved by modern algorithms in real network applications \citep{chen18}.  Lastly, although we do not show plots of \cite{chat18}'s $q_N$ and $Q_N$, these statistics exhibit similar relative performance and identical patterns as those highlighted for rSE and $cv^2$.

\begin{table}[!htbp] \centering 
	\caption{\textbf{ Simulations: $cv^2$ of estimated cardinalities for the four resolutions. } Mean $cv^2$  over $20$ realizations for $N$ in $\{5000,10000,15000\}$. The four incidence matrices are sampled once for each combination of $n$ in $\{10,20\}$ and $\mu$ in $\{5,20\}$.}
	\label{ex1cv2} 
	\scalebox{0.8}{
	\begin{tabular}{@{\extracolsep{5pt}} l |cccc} 
		\\[-1.8ex]\hline \hline
		& \multicolumn{2}{c}{n=10} & \multicolumn{2}{c}{n=20}\\
		& $\mu$=5 & $\mu$=20 & $\mu$=5 &  $\mu$=20 \\ 
		\hline \\[-1.8ex] 
		Tajima algorithm &	$0.716$ & $0.532$ & $3.674$ & $4.554$ \\ 
		Kingman algorithm &	$0.906$ & $0.698$ & $4.085$ & $3.844$ \\ 
		Tree shapes algorithm &	$1.345$ & $0.532$ & $9.068$ & $5.499$ \\ 
		Labeled trees algorithm &	$4.074$ & $1.096$ & $23.821$ & $26.815$ \\ 
		\hline \\[-1.8ex] 
	\end{tabular}}
\end{table} 

\begin{figure}[!htbp]
	\includegraphics[width=\textwidth]{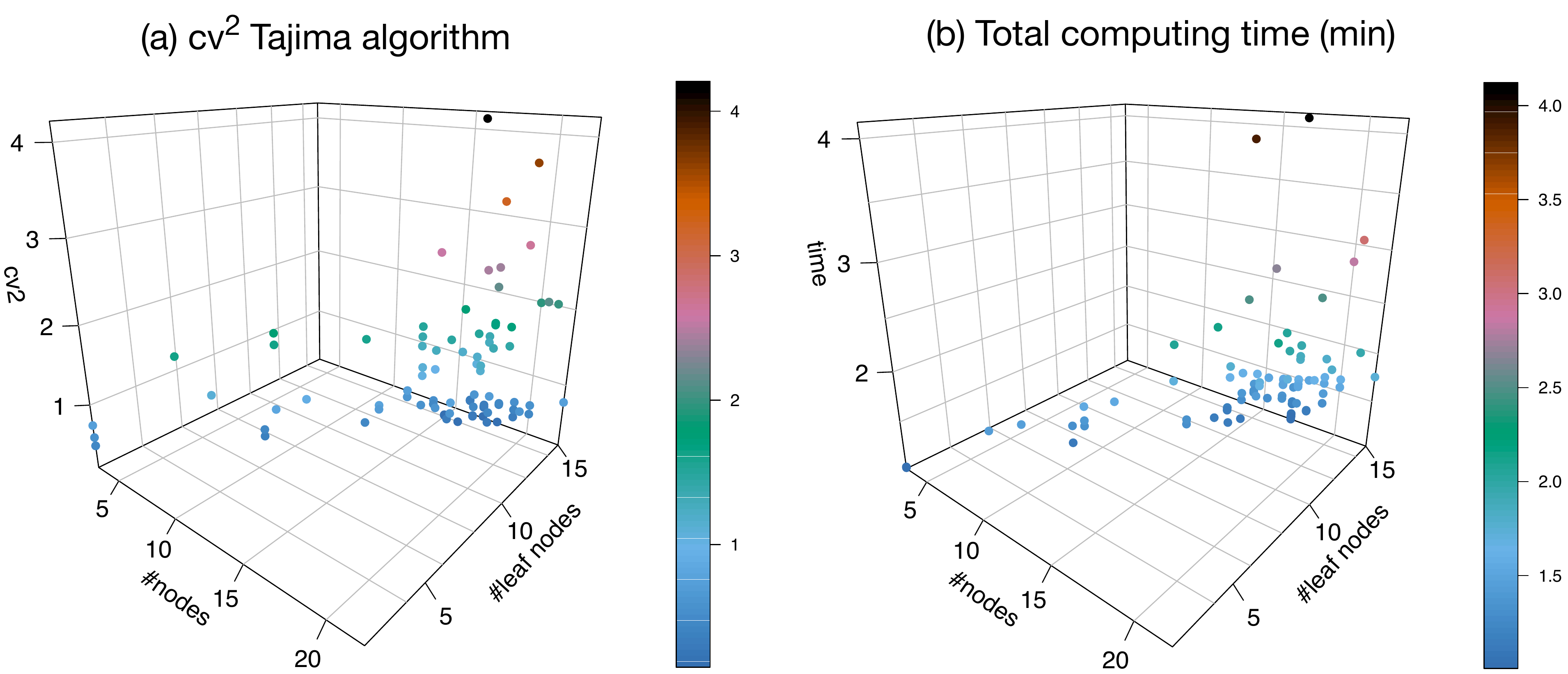}
	\caption{\small{\textbf{ Simulations: $cv^2$ of the Tajima algorithm (a) and total (all topological spaces) computing time (b).} Panel (a) plots the $cv^2$ of that Tajima algorithm as a function of the number of nodes ($\#$nodes) and the number of leaf nodes ($\#$leaf nodes) in the perfect phylogeny $\mathcal{T}^T$. Panel (b) reports the computing time to obtain the estimate counts for the four topologies and the convergence diagnostics (rSE, ESS, $q_N$, $cv^2$). Each dot corresponds to an algorithm run ($N=5000$) for each of the 100 incidence matrices simulated with fixed sample size $n=15$ and  $\mu$ in $\{1,4,7,10,13\}$.  Dot colors represent the numerical value of $cv^2$(Panel (a)) and total time in minutes (Panel (b)) as represented by the vertical bars to the left of the plots.}}
	\label{cv2/time}
\end{figure}

This first simulation study does not assess variance and scalability of our proposals for different data sets, in particular, it does not assess how the quality of the proposals depends on the perfect phylogeny $\mathcal{T}$. To study this question, we simulate $20$ incidence matrices with fixed sample size $n=15$ and for each $\mu$ in $\{1,4,7,10,13\}$ (100 matrices in total). SIS estimated cardinality and diagnostics are computed at $N=5000$.  In Figure \ref{cv2/time} (a) we show the  $cv^2$ of the estimated cardinality of ranked tree shapes (Tajima algorithm), as a function of the total number of nodes and total number of leaf nodes of $\mathcal{T}^T$. We do not plot the $cv^2$ values of the other three resolutions, however, the other resolutions mirror observed values for the Tajima algorithm. In Figure \ref{cv2/time} (b) we plot the total computing time (in minutes) to obtain the five statistics: count estimates, $cv^2$, rSE, ESS, and $Q_N$, for all topologies. In both panels of Figure \ref{cv2/time}, each dot corresponds to one incidence matrix. The mutation parameter $\mu$ is not displayed; it is varied solely to have diverse incidence matrices.

Figure \ref{cv2/time} (a) shows a non linear relationship between the $cv^2$ and both the number of nodes and leaf nodes. As expected, when the  number of nodes (leaf or total) is low, the $cv^2$ values are always low. In this case our SIS proposal distributions are close to the Kingman and Tajima jump chains, which are uniform and close to uniform distributions respectively, on the space of trees. As the number of nodes increases, the $cv^2$ values exhibit a large variation across data sets. In this case, the quality of our proposal deteriorates for few datasets substantially. Our interpretation is that more mutation constraints move our proposal distributions further away from the uniform distribution.

A second result of the simulation study is that the Kingman algorithm has a better performance, on average, than  the Tajima algorithm. The mean of the $cv^2$ values across different data sets is $1.02$ for the Tajima algorithm, and $0.72$ for the Kingman algorithm. This is coherent with the construction of our SIS proposals and the fact that the Kingman jump chain is exactly uniform. 

Lastly,  the computing time (Figure \ref{cv2/time}(b)) exhibits exactly the same patterns displayed by the  $cv^2$ (Figure \ref{cv2/time}(a)). Longer computing times are driven  mostly by the backtracking algorithm. Trivially, the more nodes in $\mathcal{T}^T$, the more nodes ordering $\textbf{v}$ the backtracking algorithm may need to explore. We observe though that the computing time remains low for most data sets, even with large number of nodes.

\section{Case studies}

\subsection{\textbf{Case Study 1}: multi-resolution simulation study}

As discussed in the introduction, there is a growing interest in population genetics to use more efficient lower resolution coalescent models for inferring evolutionary parameters from molecular data \citep{sai15,sai18,pal18}. However, no work has been done to quantify the ``real" gains of working with different coalescent resolutions to real data. It is important to address this question before this research direction is further explored. This case study addresses this question through simulations. 

\begin{table}[!htbp] \centering 
	\caption{\small{\textbf{Case study 1. Multiresolution simulation study: SIS counts for varying sample size ($n$) and mutation rate}($\mu$). $n$ denotes sample size, $\mu$ mutations rate, $|J_{leaf}|$ denotes the number of leaf nodes in $\mathcal{T}$, $|J|$ denotes the number of nodes in $\mathcal{T}$. Counts are reported for the four resolutions plus/minus the standard error.}}
	\label{ex1} 
	\scalebox{0.77}{
		\begin{turn}{0}
		{\footnotesize{	\begin{tabular}{@{\extracolsep{5pt}} cccccccc} 
				\\[-1.8ex]\hline 
				\hline \\[-1.8ex] 
				n & $\mu$ & $|J|$ & $|J_{leaf}|$ & Tajima trees & Kingman trees & Labeled trees & Tree shapes \\ 
				\hline \\[-1.8ex] 
				$5$ & $2$ & $7$ & $5$ & 4.979 +/- 0.02536 & 30.05 +/- 0.2057 & 14.97 +/- 0.1832 & 2.629 +/- 0.02622 \\ 
				& $5$ & $8$ & $5$ & 2.993 +/- 0.005794 & 8.988 +/- 0.01735 & 2.996 +/- 0.005785 & 0.9976 +/- 0.001931 \\ 
				& $10$ & $8$ & $5$ & 2.999 +/- 0.005778 & 8.99 +/- 0.01735 & 2.997 +/- 0.005784 & 0.9996 +/- 0.001926 \\ 
				& $20$ & $9$ & $5$ & 3.01 +/- 0.01414 & 3.024 +/- 0.01414 & 1.008 +/- 0.004713 & 1.003 +/- 0.004714 \\ 
				& $50$ & $9$ & $5$ & 3.016 +/- 0.01414 & 3.025 +/- 0.01414 & 1.008 +/- 0.004713 & 1.005 +/- 0.004714 \\ 
				& $75$ & $9$ & $5$ & 3.018 +/- 0.01414 & 2.976 +/- 0.01414 & 0.9919 +/- 0.004713 & 1.006 +/- 0.004714 \\ 
				$10$ & $2$ & $14$ & $9$ & 499.9 +/- 7.26 & 5350 +/- 91.37 & 21.72 +/- 0.6739 & 1.565 +/- 0.03521 \\ 
				& $5$ & $14$ & $9$ & 509.3 +/- 7.354 & 5367 +/- 91.69 & 21.94 +/- 0.6798 & 1.561 +/- 0.03405 \\ 
				& $10$ & $15$ & $9$ & 499.7 +/- 8.79 & 1765 +/- 35.78 & 7.134 +/- 0.2438 & 1.471 +/- 0.03291 \\ 
				& $20$ & $16$ & $9$ & 430.7 +/- 5.871 & 1235 +/- 17.07 & 2.94 +/- 0.04064 & 1.026 +/- 0.01398 \\ 
				& $50$ & $16$ & $9$ & 422.8 +/- 5.822 & 1235 +/- 17.07 & 2.94 +/- 0.04064 & 1.007 +/- 0.01386 \\ 
				& $75$ & $18$ & $10$ & 418 +/- 5.724 & 1249 +/- 17.35 & 2.974 +/- 0.04131 & 0.9952 +/- 0.01363 \\ 
				$15$ & $2$ & $15$ & $11$ & 3474000 +/- 53560 & 1.112e+10 +/- 141400000 & 1087000 +/- 46170 & 65.74 +/- 1.95 \\ 
				& $5$ & $20$ & $12$ & 297200 +/- 6108 & 3318000 +/- 68620 & 603.1 +/- 36.96 & 3.962 +/- 0.1131 \\ 
				& $10$ & $21$ & $12$ & 60630 +/- 1475 & 650800 +/- 14850 & 434.2 +/- 13.6 & 1.902 +/- 0.05244 \\ 
				& $20$ & $22$ & $12$ & 60330 +/- 1386 & 226000 +/- 5297 & 147.4 +/- 4.842 & 1.838 +/- 0.04675 \\ 
				& $50$ & $24$ & $13$ & 45240 +/- 961.4 & 141600 +/- 3048 & 45.3 +/- 0.9987 & 1.004 +/- 0.02134 \\ 
				& $75$ & $26$ & $14$ & 43410 +/- 894.5 & 141100 +/- 3036 & 45.1 +/- 0.9922 & 0.9637 +/- 0.01986 \\ 
				$20$ & $2$ & $20$ & $15$ & 8.05e+11 +/- 1.571e+10 & 1.731e+17 +/- 2.699e+15 & 1.588e+10 +/- 4.211e+09 & 2084 +/- 86.73 \\ 
				& $5$ & $23$ & $14$ & 8.429e+09 +/- 209800000 & 6.869e+11 +/- 1.56e+10 & 32250 +/- 1846 & 16.58 +/- 0.6664 \\ 
				& $10$ & $23$ & $13$ & 4.339e+09 +/- 83760000 & 1.189e+11 +/- 2.274e+09 & 5485 +/- 300.7 & 5.19 +/- 0.1093 \\ 
				& $20$ & $28$ & $16$ & 4.344e+09 +/- 95910000 & 2.757e+10 +/- 621200000 & 1237 +/- 85.72 & 5.102 +/- 0.1238 \\ 
				& $50$ & $28$ & $15$ & 437600000 +/- 9594000 & 2.816e+09 +/- 64850000 & 335.4 +/- 9.758 & 0.9448 +/- 0.02071 \\ 
				& $75$ & $32$ & $17$ & 458700000 +/- 10720000 & 2.754e+09 +/- 60920000 & 335.7 +/- 9.126 & 0.9904 +/- 0.02315 \\ 
				\hline \\[-1.8ex] 
			\end{tabular} }}
	\end{turn}}
\end{table} 

\begin{figure}
	\includegraphics[width=1.0\textwidth]{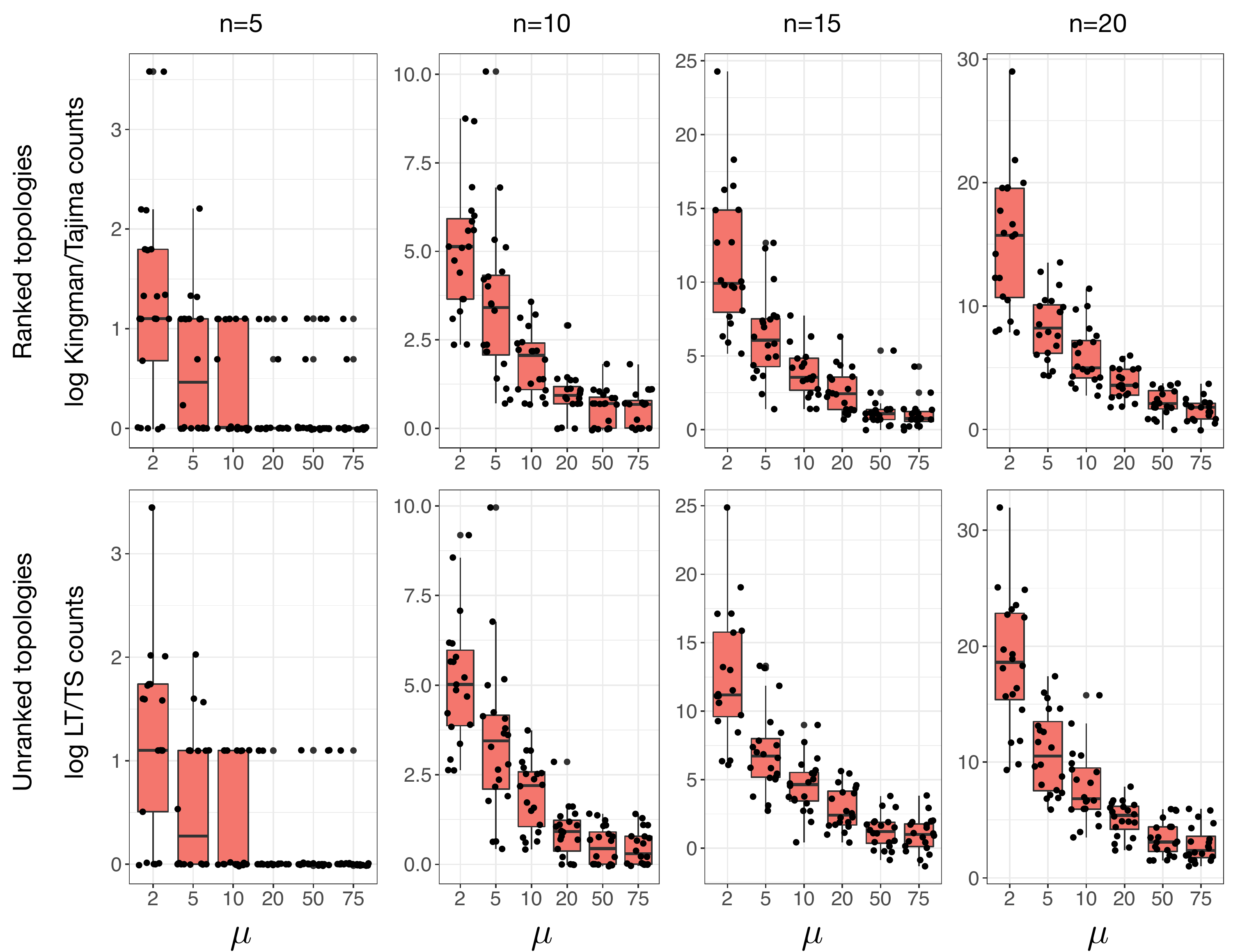}
	\caption{\small{\textbf{Case study 1. Multiresolution simulation study:} log ratio of estimated counts for varying $n$ and $\mu$.
			Rows correspond to the log ratio of cardinalities between Kingman and Tajima topologies (first row) and the log ratio of cardinalities between labeled trees and tree shapes (second row). Columns represent different sample sizes $n$ and boxplots within each plot show results under different mutation rates. Boxplots are generated from 50 independent simulations. Dots represent the SIS count estimates computed for $N= 5000$ (for $n=5,10$), $N= 10000$ (for $n=15$), and $N= 15000$ (for $n=20$). Dots are spread over the box width for ease of visualization.}}
	\label{n_mu}
\end{figure}

\textit{Data.}  We simulate $50$ incidence matrices for each of $24$ possible pairings of $n$ in $(5,10,15,20)$ and $\mu$ in $(2,5,10,20,50,75)$. For each simulated dataset, we estimate the cardinality of the four constrained topological spaces. Based on the results observed in the previous section,  we set $N=5000$ for $n \in (5,10)$, $N=10000$ for $n=15$, and $N=15000$ for $n=20$.

\textit{Results.} In the first row of Figure \ref{n_mu}, we show the log ratio of the estimated cardinalities of Kingman topologies to Tajima topologies, and in the second row, we show the log ratio of the estimated cardinalities of labeled trees to tree shapes. Table \ref{ex1} summarizes the results for a single iteration picked at random from the $50$ replicates. We note that an average over the $50$ replicates is not insightful given the high variability of the incidence matrices sampled (which can be observed by the length of the boxplots in Figure \ref{n_mu}).

Figure \ref{n_mu} shows that the cardinality of the space of Kingman trees is always larger than the cardinality of the space of Tajima trees (first row), and the cardinality of the space of labeled trees is always larger than the cardinality of the space of tree shapes (second row). As mentioned in the introduction, this is relevant for population genetic studies that aim to estimate evolutionary parameters by integrating over the space of trees. When analyzing how much effective reduction in the tree space is gained by assuming the infinite sites model alone, we note that high mutation rate will, in general, constrain the tree sample space more than a low mutation rate. This reduction is accentuated for Kingman's trees under every simulation scenario. For example for $n=20$, there are $5.64 \times 10^{29}$ (exact) unconstrained Kingman's trees (using formula from section 2). This number drops to $5.67 \times 10^{10} \pm 2.08\times 10^9$  (SIS estimate) for a simulated dataset with $\mu=20$. The (exact) unconstrained number of ranked tree shapes is $2.9 \times 10^{13}$, which drops to $4.63 \times 10^{10} \pm 1.47\times 10^8$  (SIS estimate) for a simulated dataset with $\mu=20$. A similar pattern is observed for unranked tree shapes.

For fixed sample size, we observe that the difference between Kingman and Tajima cardinalities decays exponentially from low mutation regimes to high mutation regimes (moving along $x$ axes in the plots of the first row in Figure \ref{n_mu}). The same trend is observed between labeled topologies and tree shapes (moving along $x$ axes in the plots of the second row of Figure \ref{n_mu}). For example keeping $n=20$ fixed, the Tajima space is on average approximately $2 \times 10^{11}$ smaller than the Kingman space for $\mu=2$, $6900$ times smaller for $\mu=10$, and only $7$ times smaller for $\mu=75$.

For fixed mutation regime, the respective differences between labeled topologies counts and unlabeled topologies counts, (Kingman vs. Tajima, labeled trees vs tree shapes)  becomes more pronounced as we increase the sample size (columns in Figure \ref{n_mu}.) For example keeping $\mu=10$ fixed,  the Tajima space is on average approximately $1.7$ times smaller than the Kingman space for $n=5$, $225$ times smaller for $n=10$, and only $6900$ times smaller for $n=20$.

The case study suggests that modeling with lower resolution coalescent models (unlabeled) could be advantageous when applied to organisms with low mutation rate such as humans or mammals. However, the advantages are less pronounced for rapidly evolving organisms such as pathogens and vira.

\subsection{Case study 2: Human mitochondrial and nuclear DNA data}

Present day molecular data at a nonrecombining segment (or locus) from a sample of individuals inform about past population history and other evolutionary parameters \citep{tav04}. Multiple independent loci, perhaps loci at different chromosomes or loci from distant locations across the genome provide multiple independent realizations of the same coalescent process with shared population history. 
In this case study we show that under the infinite sites model, independent chromosomal regions can impose a completely different set of constraints on their local tree topology. We provide a quantitative evidence of this effect. We note that these constraints do not arise employing alternative mutation models, e.g. Jukes-Cantor. We apply our method to mitochondrial DNA (mtDNA)  and nuclear DNA (nDNA) data. mtDNA is known to have a much higher mutation rate  than nDNA \citep{song05}, and thus, we expect a larger reduction in the state space. In addition, we explore how these constraints vary across resolutions.

\textit{Data.} We analyze $n=30$ samples of mitochondrial DNA (mtDNA) selected uniformly at random from the $107$ Yoruban individuals available in the 1000 Genomes Project phase 3 \citep{1000genomes}. We retained the coding region: $576-16,024$ according to the rCRS reference of Human Mitochondrial DNA \citep{and81,andr99} and removed 38 indels (insertions and deletions are not modeled in our approach). Of the 260 polymorphic sites, we only retained 240 sites compatible with the infinite sites mutation model. Ancestral states ($0$s in the incidence matrix) were obtained from the RSRS root sequence \citep{beh12}. For nDNA, we analyze 2320 sites of the $\beta-$ globin gene in chromosome $11$ from $n=30$ Melanesian individuals subsampled from a larger incidence matrix ($n=57$) analyzed in \citet{grif99}. It was already part of a larger dataset described in \citet{har97}. 
Figure \ref{mt} plots the Tajima perfect phylogeny for the two datasets. The mtDNA comprises  $29$ haplotypes, and the nDNA comprises $4$ haplotypes.

\textit{Results.}  We estimated the cardinalities of the four constrained topologies at $N=35000$. The number of iterations $N$ is chosen by the criteria discussed in the simulation section. Table \ref{cs2} shows the size of the unconstrained spaces, along with our SIS estimates and $cv^2$ values.

\begin{table}[!htbp] \centering 
	\caption{\textbf{ Case study 2: estimated counts and $cv^2$ for the mtDNA and nDNA datasets.} Unconstrained refers to the size of the underlying tree space (rows) cardinalities when we are not conditioning on the data. Estimates are obtained with $N=35000$, $\pm$ adds or subtract the standard error.}
	\label{cs2} 
	\scalebox{0.9}{
{\footnotesize{	\begin{tabular}{@{\extracolsep{5pt}} l | c  cc cc } 
		\\[-1.8ex]\hline 
		\hline \\[-1.8ex] 
		Dataset  (n=30) & & \multicolumn{2}{c}{Yoruban mtDNA} & \multicolumn{2}{c}{$\beta$-globin locus nDNA}\\
		\hline
		& Unconstrained & Estimate & $cv^2$ & Estimate & $cv^2$  \\ 
		\hline \\[-1.8ex] 
		Tajima  & $2.31 \times 10^{25}$ & $1.05 \times 10^{20} \pm 6.19\times 10^{18}$ & $69.2$ & $3.10 \times 10^{23} \pm 2.21\times 10^{21}$ &$1.78$\\ 
		Kingman  & $4.37 \times 10^{54}$ & $7.17 \times 10^{23} \pm 3.01\times 10^{22}$ & $36.9$ & $1.07 \times 10^{40} \pm 4.68\times 10^{37}$& $0.66$\\ 
		Tree shapes  &	$1.41 \times 10^{9}$ & $1.33 \times 10^3 \pm 1.18\times 10^2$ & $165.9$ & $3.10 \times 10^{6} \pm 2.81\times 10^5$ & $343.6$ \\ 
		Labeled trees  & $4.95 \times 10^{38}$ & $1.17 \times 10^{12} \pm 2.10\times 10^{11}$& $674.1$&  $4.65\times 10^{27} \pm 1.49\times 10^{27}$  & $372.6$ \\ 
		\hline \\[-1.8ex] 
	\end{tabular} }}}
\end{table} 

First, the space of trees compatible with the $\beta$-globin dataset are many orders of magnitude larger than with the mtDNA dataset. Following Case Study $1$, this effect could have been predicted by the lower number of segregating sites. 
Second, results in Table \ref{cs2} gives a different perspective on the computational limits of coalescent based inference: under Kingman coalescent (still the dominant model in the field of population genetics), the sample space of trees for the mtDNA is massively smaller than the unconstrained space, it drops from $4.37 \times 10^{54}$ to $7.17 \times 10^{23} \pm 3.01\times 10^{22}$. Whereas the presence of a reduction was known, such a reduction had never been quantified. 

With respect to the performance of our algorithms, we confirm that the variance of the Kingman and the Tajima algorithms (ranked tree topologies) are mostly determined by the perfect phylogeny structure rather than the sample size:  the $cv^2$ for mtDNA is larger than the $cv^2$ for nDNA (first and second rows of Table \ref{cs2}); in particular the nDNA data (4 leaf nodes) exhibits very low $cv^2$ (second row in Table \ref{cs2}). The large $cv^2$ values obtained with the unranked algorithms (third and fourth rows) questions the validity of our estimated counts. We note that the order of magnitude of the estimate is more meaningful than the point estimate itself;  the reductions in cardinality with respect to the unconstrained size are all consistent with theoretical expectations and the simulation studies. Similarly, the reductions in cardinalities across resolutions are more extreme in the mtDNA dataset than in the $\beta$-globin dataset. Surely, this case study displays a situation where a sequential importance sampler experiences variance explosion.



\begin{figure}[!htbp]
	\includegraphics[scale=0.5]{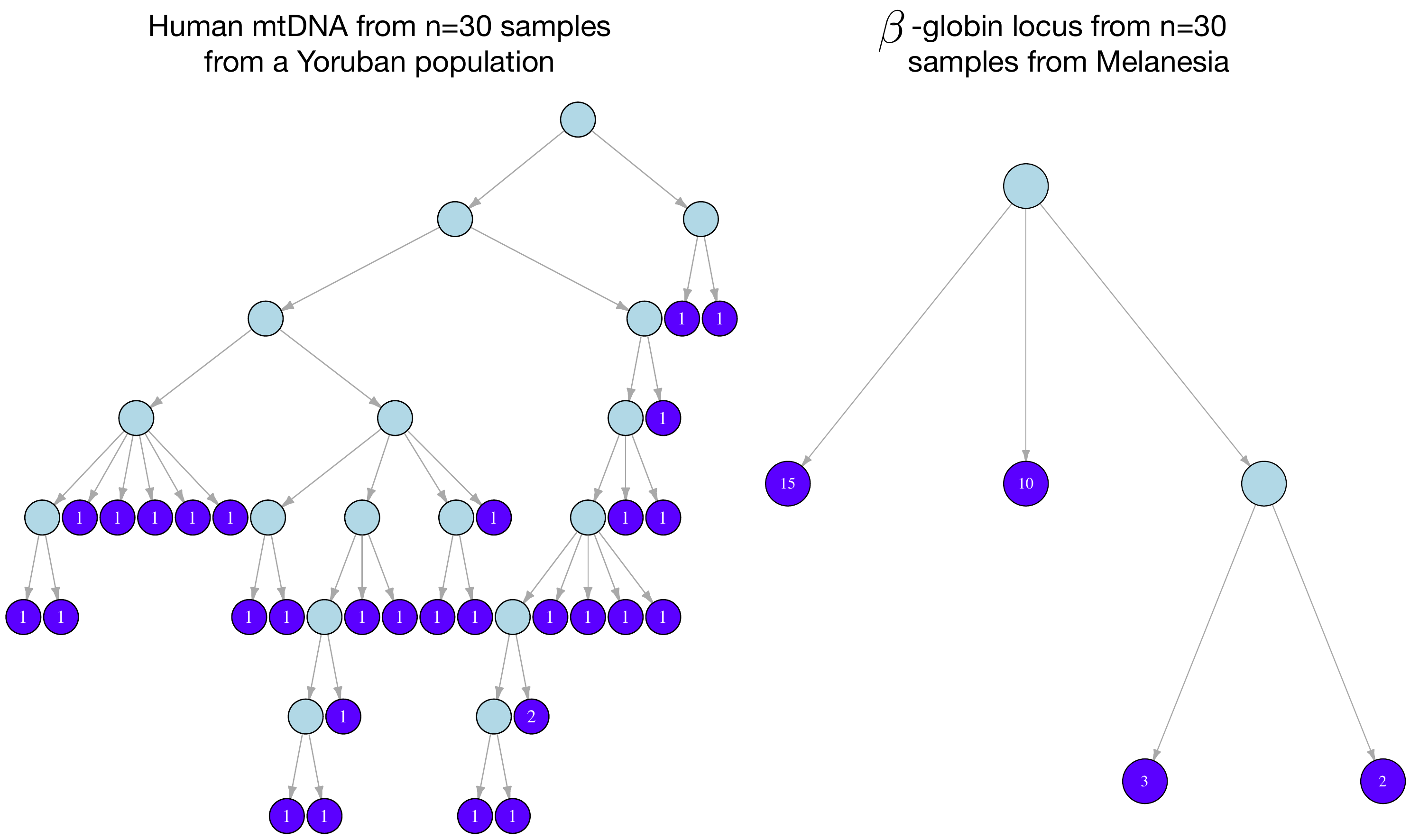} 
	\caption{\small{\textbf{Case study 2: Tajima perfect phylogenies of Yoruban mitochondrial data (left) and Melanesian $\beta-$globin locus data (right).} Left panel: $\mathcal{T}^T$ of  $n=30$ sequences of mtDNA  sampled at random from 107 Yoruban individuals available from the 1000 Genomes Project phase 3 \citep{1000genomes}. Right panel: $\mathcal{T}^T$ of  $n=30$ sequences of DNA from the $\beta$-globin locus sampled at random from 57 Melanesian individuals available  in \cite{ful94}. Dark blue nodes represent the leaf nodes. The number within a node is the number of individuals assigned to that node.}}
	\label{mt}	
\end{figure}

\section{Discussion}

In this article, we propose a set of algorithms to sequentially sample tree topologies compatible with the observed data. We use our sampling algorithms to estimate the cardinality of the sample space of tree topologies with importance sampling. We assume that our sampled locus is nonrecombining and that the infinite sites assumption holds. In the infinite sites mutation model, each site in the locus can mutate only once. While in practice it is possible to observe sites that are not compatible with this mutation model, the percentage of these cases is usually marginal for some organisms such as humans and other primates. The major implication of the infinite sites mutation model is that observed data impose constraints on the space of compatible trees. We analyze the cardinality of the following constrained tree spaces: ranked labeled trees (Kingman), ranked tree shapes (Tajima), unranked labeled  trees and tree shapes. These sample tree spaces  correspond to different resolutions of the $n$-coalescent process. 

Our proposed algorithms sample a tree topology in a bottom-up fashion: given a sample of $n$ individuals, we sequentially build the trees in $n-1$ steps. We employ a graphical representation of the data called perfect phylogeny that allows us to account for the combinatorial constraints imposed by the data. The perfect phylogeny ``groups'' individuals in different nodes: in our algorithms coalescent events are allowed solely among individuals assigned to the same node. Within each node, the choice of which individuals coalesce is regulated by the underlying jump chain of the coalescent process we are modeling. 


The research question tackled in this paper was motivated by the challenging inference problem of coalescent methods used in population genetics.  There is a growing interest in exploring different resolutions of the $n$-coalescent process for inference of evolutionary parameters from molecular sequence data in order to gain computational tractability.  
Indeed, the size of the hidden state-space of trees in the standard Kingman coalescent grows superexponentially with the sample size.  Despite the a priori reduction in the cardinality of the state-space obtained by using coarser modeling resolutions, e.g. Tajima $n$-coalescent, a quantification of this reduction  conditionally on the data was unknown. Given the amount of work and software available tailored to the Kingman $n$-coalescent, it was in our opinion fundamental to quantify the benefits of modeling with different resolutions before any more work is carried out.

From our empirical analyses, it emerges that the benefits of using a coarser resolution depends largely on the data considered. The advantages are striking as the sample size increases, especially in regimes of low mutation rate such as in nuclear human DNA variation. In general, the greater the number of observed mutations is, the less are the benefits of employing coarser resolutions. This is consistent with theoretical predictions: under the infinite sites assumption, mutations induce some labeling: individuals can be distinguished according to private mutations. In this case, the benefits of employing an unlabeled tree are less evident. This observation applies to both ranked and unranked trees. In applications where the number of mutations is  low, the benefits of coarser resolutions remain clear.

In the context of recombination, the perfect phylogeny is not longer a single tree but a set of trees called \textit{perfect phylogeny forest} \citep{gus14} and we believe that our methodology can be extended in this context. Indeed, many new interesting research questions open up; for example, the number of trees in the forest, i.e. the number of recombination events, which is itself a challenging problem known as the \textit{minimum perfect phylogenetic forest problem}. In this case, the target is not a space of tree genealogy but a space of networks known as \textit{the ancestral recombination graph} \citep{griffiths1997ancestral} and a future area of research.

\section*{Acknowledgments}
We would like to acknowledge Persi Diaconis who brought our attention to the use of sequential importance sampling for approximate counting. This work is supported by R01 GM131404 and the Alfred P. Sloan Foundation. We would like to acknowledge two anonymous reviewers for their suggestions that greatly improved the manuscript.

\bibliographystyle{agsm}
\bibliography{biblio_postdoc}

\end{document}